\documentclass[12pt]{article}
\usepackage{amscd, amsmath}
\usepackage{amsthm, amssymb, mathrsfs}
\usepackage{mathtools}
\usepackage{sectsty}
\usepackage{slashed}
\usepackage{scalerel}
\usepackage[bottom,flushmargin]{footmisc}
\usepackage{mathrsfs}
\usepackage{float}
\usepackage{array}
\usepackage{accents}
\usepackage{empheq}
\usepackage[dvipsnames]{xcolor}
\usepackage[most]{tcolorbox}
\usepackage[british]{babel}
\usepackage{tikz-cd} 
\usepackage{hyperref}
\usepackage{tocloft}
\usepackage{titlesec}
\usepackage{cite} 
\usepackage{lmodern}
\usepackage{enumitem}
\usepackage{appendix}

\titlelabel{\thetitle.\quad}
\titleformat*{\section}{\fontsize{16}{19}\bfseries\selectfont}
\titleformat*{\subsection}{\fontsize{13}{17}\bfseries\selectfont}
\subsubsectionfont{\normalfont\itshape}

\renewcommand{\baselinestretch}{1.2}
\newcommand{\linebr}{\\[0.2em]}

\DeclareMathAlphabet{\mathpzc}{OT1}{pzc}{m}{it}

\definecolor{myblue1}{RGB}{0, 0, 139}

\hypersetup{
    linktocpage,
    colorlinks,
    citecolor=myblue1,
    filecolor=black,
    linkcolor=myblue1,
    urlcolor=black
}

\newtheorem{theorem}{Theorem}[section]
\newtheorem{lemma}[theorem]{Lemma}
\newtheorem{proposition}[theorem]{Proposition}
\newtheorem*{proposition*}{Proposition}
\newtheorem{definition}{Definition}[section]

\newtheorem{corollary}[theorem]{Corollary}

\newtheorem*{assumption(A)}{Condition (A)}

\theoremstyle{remark}
\newtheorem{remark}{{\bf Remark}}[section]
\newtheorem*{remarknotation}{{\bf Remark on notation}}

\numberwithin{equation}{section}

\newcommand{\Lie}{\mathcal L}
\newcommand{\VV}{\mathscr V}
\newcommand{\HH}{\mathcal H}

\newcommand{\cL}{\mathscr{L}}

\newcommand{\Cl}{\mathrm{Cl}}
\newcommand{\CC}{\mathbb{C}}

\newcommand{\dd}{{\mathrm d}}
\newcommand{\vol}{\mathrm{vol}}
\newcommand{\Vol}{\mathrm{Vol}}
\newcommand{\diag}{\mathrm{diag}}

\newcommand{\Diff}{{\rm Diff}}

\newcommand{\divergence}{\mathrm{div}}
\newcommand{\gradient}{\mathrm{grad}}

\newcommand{\sD}{\slashed{D}}

\newcommand{\beq}{\begin{equation}}
\newcommand{\eeq}{\end{equation}}
\newcommand{\bal}{\begin{align}}
\newcommand{\eal}{\end{align}}
\newcommand{\bmatr}{\begin{bmatrix}}
\newcommand{\ematr}{\end{bmatrix}}

\newcommand{\overbar}[1]{\mkern 1.5mu\overline{\mkern-1.5mu#1\mkern-1.5mu}\mkern 1.5mu}

\newcommand{\bz}{\bar{z}}
\newcommand{\blangle}{\boldsymbol{\langle}}
\newcommand{\brangle}{{\boldsymbol{\rangle}}}
\newcommand{\Blangle}{\boldsymbol{\big\langle}}
\newcommand{\Brangle}{{\boldsymbol{\big\rangle}}}

 \newtcolorbox{empheqboxed}{ 
 opacityback=0,
 enhanced jigsaw,
 width=\textwidth,
 boxrule=.5pt,
 sharpish corners,
 left=0pt,
 right=2pt,
 top=-9pt, 
 bottom=3pt
}

\newcommand{\PPP}{{\text{\scalebox{.8}{$P$}}}}
\newcommand{\MMM}{{\text{\scalebox{.8}{$M$}}}}
\newcommand{\KKK}{{\text{\scalebox{.8}{$K$}}}}

\DeclareSymbolFont{TOneChars}{T1}{\familydefault}{m}{it}
\SetSymbolFont{TOneChars}{bold}{T1}{\familydefault}{bx}{it} 
\DeclareMathSymbol{\mathdh}{\mathord}{TOneChars}{"F0}

\makeatletter
\def\blfootnote{\xdef\@thefnmark{}\@footnotetext}
\makeatother

\begin{document}

\begin{titlepage}

\title{{\bf Chiral interactions of fermions and \\massive gauge fields  in Kaluza-Klein models\vspace{.5cm}}}

\author{Jo\~ao Baptista} 
\date{\vspace{-.3cm}June 2025}

\maketitle

\thispagestyle{empty}
\vspace{1cm}
\vskip 10pt
{\centerline{{\large \bf{Abstract}}}}
\noindent
In Kaluza-Klein theory, gauge fields on $M_4$ arise as components of a higher-dimensional metric defined on $M_4 \times K$. The traditional expectation is that all the gauge fields of the Standard Model are linked to exact Killing vector fields on the internal space. This paper questions that assumption and investigates the properties of 4D gauge fields linked to non-Killing fields on $K$. It is shown that they have massive yet arbitrarily light bosons; they can mix fermions with different masses; and they can have asymmetric couplings to left- and right-handed fermions. None of these properties is easily satisfied by gauge fields linked to internal isometries. Thus, backgrounds encoding such massive gauge fields circumvent traditional no-go arguments and offer a geometric source of chiral interactions with fermions. This may help to model the weak force within the Kaluza-Klein framework. Technically, the paper uses the language of spin geometry and Riemannian submersions. It studies the higher-dimensional Dirac operator with non-trivial background metrics. The results are derived for a general $K$. They are illustrated explicitly in the simpler cases where $K$ is the two-sphere and the two-torus.

\vspace{-1.0cm}
\let\thefootnote\relax\footnote{
\noindent
{\small {\sl \bf  Keywords:} Kaluza-Klein theories; chiral fermions; Riemannian submersions; Dirac operator; weak force.}
}

\end{titlepage}

\pagenumbering{roman}

\renewcommand{\baselinestretch}{1.1}\normalsize
\tableofcontents
\renewcommand{\baselinestretch}{1.19}\normalsize

\newpage

\pagenumbering{arabic}

\section{Introduction and overview of results}
\label{Introduction}

The weak force is a special interaction in the Standard Model. Notable properties that distinguish it from the strong and electromagnetic forces include:
\begin{enumerate}[itemsep=1mm, topsep=7pt]
\item It has massive gauge bosons;
\item It mixes fermions with different masses;
\item It has asymmetric (chiral) couplings to left- and right-handed fermions.
\end{enumerate}
This paper describes how, in a general Kaluza-Klein framework, these three properties of 4D gauge fields are manifestations of a single higher-dimensional feature. So the three properties should generally appear together. Massive gauge bosons, however light, should have chiral interactions with fermions, whereas massless gauge bosons should not.

To introduce that higher-dimensional feature, recall that Kaluza-Klein theory encodes 4D gauge fields as components of a metric defined on $P=M_4 \times K$. The usual assumption is that all the gauge fields of the Standard Model are linked to isometries of the internal vacuum metric. In other words, they are linked to exact Killing vector fields on $K$ \cite{Witten81, Duff, Bailin, CJBook, CFD, WessonOverduin}. This paper probes that assumption and shows that the 4D gauge fields linked to non-Killing vector fields will in general exhibit properties 1, 2 and 3, while the gauge fields linked to isometries of $K$ will not. So the standard view appears well-founded in the case of the strong and electromagnetic fields. However, it may be preferable to link the weak gauge field to non-Killing vector fields on $K$ that are small perturbations of Killing ones.

Non-Killing vector fields generate non-isometric diffeomorphisms of the internal space. These transformations do not preserve the vacuum metric $g_K$ but do preserve the Einstein-Hilbert action, which is diffeomorphism-invariant. This is reminiscent of the standard Brout-Englert-Higgs mechanism \cite{EnBr, GHT, Higgs}, which uses gauge transformations that preserve the action but do not preserve the vacuum value of the Higgs field. The internal metric $g_K$ can therefore be understood as a more geometrical version of the traditional Higgs fields. This is apparent from the decomposition of the higher-dimensional scalar curvature $R_{g_\PPP}$ in the Einstein-Hilbert action, 
 \beq \label{GaugedSigmaModelAction0}
\int_P \, R_{g_\PPP}  \vol_{g_\PPP}   \ = \ \int_P \, \Big[\, R_{g_\MMM}  \, + \, R_{g_\KKK} \, - \, \frac{1}{4}\, |F_A|^2 \, - \,  \frac{1}{4}\,  |\dd^A  g_K|^2  \, + \,  |\dd^A \, (\vol_{g_\KKK})|^2 \, \Big] \,\vol_{g_\PPP} \, . 
\eeq
 This formula extends the usual Kaluza-Klein result to the setting of general Riemannian submersions, where the geometry of the fibres can change. For more details, see \cite{Bap, Besse}. 
 
For a general submersion metric on $P$, the gauge one-form $A_\mu$ has values in the full space of vector fields on $K$, which is the Lie algebra of the diffeomorphism group ${\rm Diff} (K)$. It can be expanded as $A^a_\mu\, e_a$, where $\{ e_a \}$ is a set of independent vector fields on $K$, some of them Killing and others non-Killing. The gauge group is ${\rm Diff} (K)$ or a subgroup. It need not act on $K$ only through isometries of $g_K$.
 As in \cite{Bap}, one can then calculate that the 4D gauge boson linked to a divergence-free $e_a$ has a classical mass given by 
\beq   \label{MassFormula}
\left(\text{Mass} \ A_\mu^a \right)^2 \ \ \propto \ \ \frac{ \int_K  \; \left\langle \Lie_{e_a}\, g_K,  \; \Lie_{e_a}\, g_K \right\rangle  \, \vol_{g_\KKK} }{ 2 \int_K  \,  g_K (e_a ,  e_a  ) \ \vol_{g_\KKK} }  \ ,
\eeq
where $\Lie_{e_a} g_K$ denotes the Lie derivative of the internal vacuum metric along $e_a$. 

Formula \ref{MassFormula} is significant. It suggests that massive gauge fields, however light, should not be linked to exact isometries of $g_K$. The derivatives $\Lie_{e_a} g_K$ can be small yet non-zero. Thus, a natural way to obtain light gauge bosons within the Kaluza-Klein framework is to start with an internal metric with a large isometry group, for example an Einstein metric, and assume that a perturbation operating at a different scale\footnote{Presumably, a higher-order correction to the Einstein-Hilbert action, or a quantum effective potential. In the usual Higgs mechanism, the ad-hoc Higgs potential forces the shift of the vacuum.} slightly shifts the vacuum metric and breaks its isometry group. Some Lie derivatives $\Lie_{e_a} g_K$ will then become non-zero and the corresponding gauge fields will become massive. If the perturbation is small, the masses will be small too. Once those gauge fields are linked to non-Killing vector fields on $K$, they can satisfy conditions 1, 2 and 3.

Among the three listed properties of gauge fields linked to non-Killing internal fields, the least predictable one is perhaps the emergence of chiral interactions with fermions. In fact, the literature has several no-go arguments against the existence of chiral fermions in Kaluza-Klein models. Firstly, the basic Kaluza-Klein mechanism to explain 4D chiral fermions relies on the correlation between internal and 4D chiralities, which only works when $K$ is even-dimensional \cite{Manton, CS, Witten83}. But an analysis of conjugation in the spin representations excludes internal spaces of dimension $4k$ \cite{CS, Wetterich83a, Wetterich83, Witten83}. Additionally, a prominent no-go argument presented by Witten in \cite{Witten83} excludes all internal spaces with even dimensions. This argument is based on an index result of Atiyah and Hirzebruch \cite{AH}, so remains valid even if the Dirac operator is continuously deformed, for instance through the addition of torsion to the metric connection.

Nonetheless, an inspection of these no-go arguments reveals a common assumption: that all 4D gauge fields of interest are linked to exact Killing vector fields on $K$. Thus, gauge fields linked to non-Killing fields on $K$ have the potential to evade the arguments, as noted in \cite{Bap}. 
Even light 4D gauge fields, linked to very small perturbations of internal Killing fields, should be able to evade the no-go arguments. This paper confirms this is indeed the case. It shows that arbitrarily light gauge fields have chiral interactions with fermions. Explicit examples are given when $K=S^2$ and $K=T^2$.

There are other well-known proposals in the literature to circumvent the no-go arguments against chiral fermions in Kaluza-Klein. For example, one can add to the model gauge fields  living on the higher-dimensional spacetime \cite{CM, Manton, CS, RSS, FY, CIF}. The additional fields will twist the internal Dirac operator and can lead to the appearance of fermionic zero modes with chiral interactions. One can also use non-compact internal spaces \cite{Wetterich84}. One can construct models using generalized versions of Riemannian geometry \cite{Weinberg84, Wetterich84, Moffat}. One can consider metric connections with topologically non-trivial torsion \cite{Neville, Tchrakian}. And a fifth, prominent strategy is to abandon the smoothness assumption and consider vacuum internal spaces with singularities, for example orbifolds. Then one can impose chiral boundary conditions at the singularity and obtain models with 4D chiral fermions \cite{DHVW, DHVW2, PQ, DDG, CDH, GGH, Sundrum, CHM}. 

To this author's prejudice, however, the mechanism proposed in the present paper seems somewhat more natural. It does not introduce any new bosonic fields besides the higher-dimensional metric, in accordance with the original Kaluza-Klein philosophy. Chiral interactions follow from the properties of the standard Dirac operator in higher-dimensions, with no need for more complicated operators. It uses smooth, compact internal spaces and the ordinary version of Riemannian geometry. And, importantly, the fermionic interactions generated in this manner can be chiral only for gauge fields with non-zero mass, as observed in nature. 

Extended reviews of the Kaluza-Klein framework can be found in \cite{Bailin, Bou, CFD, CJBook, Duff, WessonOverduin, Witten81}. Some of the early original references are \cite{Kaluza, Klein, EB, Jordan, Thiry, DeWitt, Kerner, Cho}. This paper follows the notation and treatment given in \cite{Bap, Bap2} of massive gauge fields and Higgs-like scalars. A clear limitation is that the analysis is only classical, all throughout. More comments about the literature will be added below, as we give an overview of the main results in the paper.

\subsection*{Overview of the general results}

Section \ref{Submersions} starts by describing properties of spinors on the manifold $P= M_4 \times K$ equipped with a metric $g_P \simeq (g_M, A, g_K)$ defining a Riemannian submersion. These metrics generalize the Kaluza ansatz by encoding not only the 4D metric and 4D massless gauge fields, but also massive gauge fields and an internal metric $g_K$ that can change along $M_4$. The main classical results about Riemannian submersions were developed by O'Neill in \cite{ONeill1}, after foundational work in \cite{Ehresmann, Hermann}. They are presented in \cite{Besse, FIP}, for example. We use the translation of those results to the Kaluza-Klein language provided in \cite{Bap}, which considers the case of general fibres that need not be totally geodesic. Spinors on general Riemannian submersions have been previously studied by Moroianu \cite{Moroianu} and, later, by Loubeau and Slobodeanu \cite{LS} and Reynolds \cite{Reynolds} for general conformal submersions. Section \ref{Submersions} reframes that geometric setting in the explicit language suitable for Kaluza-Klein physics. 

Section \ref{SectionDecompositionDiracOperator} uses that same language to describe the covariant derivatives of spinors over Riemannian submersions. They split into derivatives along vertical fields (vector fields on $K$) and along horizontal  fields (lifts of vector fields on $M_4$). These two cases come together in the expression for the standard Dirac operator on $P = M_4 \times K$. This formula is spelled out in proposition \ref{SimplerDecompositionDiracOperator}, for the case of a constant internal metric $g_K$, and in proposition \ref{thm:DecompositionDiracOperator} for the general case. Although mathematically equivalent to a special case of the formulae given by Reynolds in \cite{Reynolds}, the expressions given in these propositions look quite different. They have two advantages for our purposes. First, they use the physical language of Kaluza-Klein models. Second, by explicitly developing certain terms left unexpanded in \cite{Reynolds}, they show that general gauge fields couple to 4D spinors through the Kosmann-Lichnerowicz derivative $\cL_X$ of the internal component of those spinors. This fact is well-known in the physics literature for gauge fields linked to internal Killing vector fields \cite{Witten83, Wetterich83}. The generalization to gauge fields linked to a non-Killing $X$ over a general Riemannian submersion does not seem to have been established yet. Thus, very schematically, one can write
\beq
\label{GeneralFermionCoupling}
\sD^P \Psi \; = \; \sum_{\mu=0}^3\,  \Gamma^\mu\, (\partial_\mu + A^a_\mu \, \,\cL_{e_a})\, \Psi \; + \; \sD^K \, \Psi \; +\; \cdots \ ,
\eeq
where $\Psi$ is a higher-dimensional spinor and $\sD^P$ and $\sD^K$ denote the standard Dirac operators on $P$ and $K$. When $\Psi$ is an eigenspinor of the internal Dirac operator, the term $\sD^K \, \Psi$ produces a mass term in four-dimensions, as is customary in Kaluza-Klein models. The gauge fields have values in  the Lie algebra of vector fields on $K$. So the $\{ e_a \}$ form a basis of such vector fields. Some of them will be Killing and others non-Killing.

Due to its appearance in decomposition \eqref{GeneralFermionCoupling}, the Kosmann-Lichnerowicz derivative $\cL_X$ plays a major role in this story. So we take section \ref{PropertiesKLD} to collect useful properties of this operator for non-Killing $X$. Almost all of them can be found in the literature in the Riemannian case \cite{Kosmann, BT}, but are somewhat dispersed and stated using various conventions. One property that is especially relevant for us is the commutator $[\cL_X, \sD^K]$. Given a vector field $X$ and a spinor $\psi$ on $(K, g_K)$, that commutator is
\begin{align}
[\sD^K\!,\,  \cL_X] \, \psi  \ = \ & \frac{1}{2}\, \sum_{i, j} \, (\Lie_X g_K) (v_i , v_j) \ v_i \cdot \nabla_{v_j} \psi  \ + \linebr
& + \frac{1}{4}\, \sum_{i, j}  \, \big\{ [\nabla_{v_i}(\Lie_X g_K)] (v_i , v_j)  \; - \;    [\nabla_{v_j}(\Lie_X g_K)] (v_i , v_i)  \big\} \; v_j \cdot \psi  \ . \nonumber 
\end{align}
Here $\nabla$ denotes the Levi-Civita connection and $\Lie_X g_K$ is the Lie derivative of $g_K$ along $X$. The $\{v_j\}$ are a local, orthonormal trivialization of the tangent bundle $TK$. So the commutator vanishes when $X$ is a Killing vector field, but in general does not. In light of \eqref{GeneralFermionCoupling}, this says that a gauge field $A_\mu^a$ will mix fermions with different masses only when the respective vector field $e_a$ is not Killing on $K$. This is one reason why the weak gauge field should be linked to non-Killing internal vector fields.

Now assume that $K$ is even-dimensional and let $\Gamma_K$ denote the complex chirality operator on $K$, normalized so that $(\Gamma_K)^2 = 1$. We can decompose the complex spinor bundle as $S_\CC(K) = S_\CC^+(K)  \oplus S_\CC^-(K) $ and write each spinor as a sum of chiral components,
\beq
\label{DefChiralProjections}
\psi \ = \ \psi^+ \; + \;  \psi^- \qquad \quad {\rm with} \quad \   \psi^\pm \; =\; \frac{1}{2} \, (1 \pm \Gamma_K) \, \psi \ .
\eeq
The Kosmann-Lichnerowicz derivative $\cL_X$ always commutes with $\Gamma_K$. So it preserves the spaces of sections of the half-spinor bundles,  denoted $V^\pm$. These are infinite-dimensional spaces of Weyl spinors. The question we want to study in section \ref{ChiralFermions} is whether $\cL_X$ acts similarly on $V^+$ and $V^-$. Using the standard assumption of correlation between 4D and internal chiralities \cite{Witten83}, that is equivalent to asking whether the gauge field linked to $X$ acts symmetrically on left- and right-handed 4D spinors.  

When $\cL_X$  acts similarly on $V^+$ and $V^-$, we say that it has chiral symmetry. This notion should imply an identity of matrix elements $\langle \phi^+, \cL_X \psi^+ \rangle = \langle \phi^-, \cL_X \psi^- \rangle$ for all $\sD^K$-eigenspinors $\phi$ and $\psi$ with positive eigenvalues. In fact, in section \ref{ChiralFermions} we will distinguish two notions of chiral symmetry: when that identity holds exactly (strong symmetry), and when it holds only up to a unitary redefinition of the $\sD^K$-eigenspinors (weak symmetry). In the latter case, the redefinitions must respect the mass eigenvalues and the symmetries of spinors induced by the isometry group of $g_K$. 

Let us probe the main question. Consider two $\sD^K\!$-eigenspinors $\phi_{m'}$ and $\psi_{m}$ with positive eigenvalues $m'$ and $m$. The Dirac operator anti-commutes with $\Gamma_K$ and is formally self-adjoint with respect to the $L^2$-inner-product of spinors. The operator $\Gamma_K$ is self-adjoint with respect to that same product. So a simple calculation yields
\begin{align}
\label{RelationCommutatorChirality}
\int_K \langle\, \phi_{m'}, \, [\sD^K\!,\,  \cL_X] \, \Gamma_K  \psi_m \, \rangle \, \vol_{g_\KKK}\, &= \, (m+m') \int_K  \langle\,  \Gamma_K \phi_{m'}, \; \cL_X  \psi_m \, \rangle  \, \vol_{g_\KKK}  \linebr
&= \, (m+m') \int_K \big\{ \langle\, \phi_{m'}^+ , \; \cL_X \psi_m^+  \, \rangle   -   \langle\, \phi_{m'}^- , \; \cL_X \psi_m^-  \, \rangle \big\} \, \vol_{g_\KKK} \, . \nonumber
\end{align}
When $X$ is Killing on $K$, the commutator $ [\sD^K\!,\,  \cL_X]$ vanishes and the matrix elements of $\cL_X$ are the same on left- and right-handed spinors. When $X$ is not Killing, $\cL_X$ should in general act differently on $\psi_m^+$ and $\psi_m^-$. So it will not have strong chiral symmetry.

The traditional no-go arguments against the existence of chiral fermions in Kaluza-Klein models are encapsulated in the following important result:
\begin{proposition}[\!\!\cite{Witten83, AH}]
\label{thm: KillingCase}
If $X$ is a Killing vector field on a compact, connected, even-dimensional spin manifold $K$, the derivative $\cL_X$ has strong chiral symmetry.\footnote{This result does not follow directly from \eqref{RelationCommutatorChirality} because that formula is uninformative when $m=m'=0$.}
\end{proposition}
This symmetry, together with the correlation between internal and 4D chiralities, forces all massless gauge fields to have the same couplings to left- and right-handed 4D fermions. 
The main message of the present paper, in contrast, is that this property of $\cL_X$ is specific to Killing fields. It should not be generic. It will not hold for most non-Killing vector fields on $K$, even if they are small perturbations of Killing fields. In fact, the explicit calculations on the sphere described in section \ref{ExplicitExampleS2} show that:
\begin{proposition}
\label{thm: SphereCounterexample}
Let $X$ be a generic divergence-free vector field on  $K=S^2$ equipped with its unique spin structure. Then the derivative $\cL_{X}$ of spinors does not have strong chiral symmetry.
\end{proposition}
The calculations on the torus described in section \ref{ExplicitExampleT2} go further. They show that it is not possible to rotate the eigenspinors of $\sD$ within their eigenspaces so that the action of $\cL_{X}$ on $V^+$ looks similar to its action on $V^-$.
\begin{proposition}
\label{thm: TorusCounterexample}
Let $X$ be a generic Hamiltonian vector field on the torus $K=T^2$ equipped with its trivial spin structure. Then the derivative $\cL_{X}$ has neither strong nor weak chiral symmetries.
\end{proposition}
The precise definitions of strong and weak chiral symmetries are given in section \ref{InternalChiralFermions}. The fact that $\cL_{X}$ does not have such symmetries for a non-Killing $X$ should hold much more generally, beyond the simple cases illustrated in this paper.

Thus, the broad picture provided in this paper is that the solutions of the equation $\sD^P \Psi = 0$ do not involve chiral fermionic interactions in regions where the background metric is a simple product, $g_P = g_M + g_K$, or where it is a submersion metric $g_P \simeq (g_M, A, g_K)$ with constant $g_K$ and only massless gauge fields. In regions where massive gauge fields are present, in contrast, the background geometry becomes very different from a product and, in general, chiral fermionic interactions will emerge.

\subsection*{Explicit example: $K=S^2$}

Choose $K=S^2$ equipped with its round metric and unique spin structure. Its space of spinors has an infinite basis formed by eigenspinors $\psi_{l,\, n}$ of the internal Dirac operator. Each eigenspinor determines the properties of the 4D fermion $\varphi$ coupled to it in the tensor product $\Psi = \varphi \otimes \psi_{l,\, n}$, which is a spinor over $M_4\times S^2$. The half-integer $l$ accounts for the $\sD^K\!$-eigenvalue of $\psi_{l,\, n}$. It determines the 4D mass of $\varphi$. The half-integer $n$ is in the range $\{-l, -l+1, \ldots , l \}$ and reflects the eigenvalue of $\psi_{l,\, n}$ under the action of the internal operator $\cL_{\partial_\phi}$. Here $\partial_\phi$ is the Killing vector field generating azimuthal rotations on $S^2$. That eigenvalue measures the charge of $\varphi$  when responding to the massless 4D gauge field linked to $\partial_\phi$, as follows from \eqref{GeneralFermionCoupling}. Thus, the only characteristic numbers of 4D fermions in this model are their mass and $\partial_\phi$-charge.

Now suppose that a 4D gauge field is linked to an internal, divergence-free vector field $X$ on $S^2$. On the sphere, $X$ is necessarily the Hamiltonian vector field $X_h$ of some real function $h$. This simplifies the calculation of $\cL_{X_h}$. For example, when acting on the internal spinor $\psi_{\frac{1}{2},\, \frac{1}{2}}$, a direct computation in section \ref{ChiralFermionsS2} says that
\begin{multline}
\label{ChiralExampleS2}
\int_{S^2} \big\{ \langle\, \psi_{l,\, n}^+ \, , \; \cL_{X_h} \psi_{\frac{1}{2},\, \frac{1}{2}}^+  \, \rangle   -   \langle\, \psi_{l,\, n}^- \, , \; \cL_{X_h} \psi_{\frac{1}{2},\, \frac{1}{2}}^-  \, \rangle \big\} \, \vol_{g_\KKK}  \; = \;  \linebr 
=\; \frac{i}{8}\, \sqrt{\frac{l+ n}{2\pi\,  l}} \, \Big(l- \frac{3}{2} \Big) \, \Big(l- \frac{1}{2} \Big)  \int_{S^2}\, \overbar{Y_{l - \frac{1}{2},  \,  n - \frac{1}{2}}^0}  \,\,  h \, \, \vol_{g_\KKK} \ .
\end{multline}
Here $Y_{l - \frac{1}{2},  \,  n - \frac{1}{2}}^0$ denotes the scalar spherical harmonic on $S^2$. This formula shows that the massive 4D gauge field linked to  $X_h$ will generically have chiral interactions with the lightest 4D fermion, i.e. with the fermion $\varphi$ appearing in the higher-dimensional spinor $ \varphi \otimes \psi_{\frac{1}{2},\, \frac{1}{2}}$. The amount of chirality depends on $X_h$ through the harmonic components of the Hamiltonian function $h$. One can check that for the Killing choice $X_h = \partial_\phi$ the right-hand side of \eqref{ChiralExampleS2} always vanishes. Note also that chirality is manifest only in the higher mass components of $\cL_{X_h} \, \psi_{\frac{1}{2}, \, \frac{1}{2}}$, namely, in the components with $l \geq 5/2$. Returning to the Kaluza-Klein model over $M_4 \times S^2$, this means that the lightest 4D fermions can have chiral interactions only when higher mass fermions are also involved, not when the interaction involves only the lightest fermions.

\subsection*{Explicit example: $K=T^2$}

The simplest example to calculate is when $K$ is the two-torus $T^2 = \mathbb{R}^2 / \mathbb{Z}^2$ with its flat metric and trivial spin structure. Then the Dirac operator $\sD^K$ has eigenspinors $\psi_{l_1,\,  l_2}$  characterized by two integers $(l_1, l_2) \in  \mathbb{Z}^2$. 
The spinors $\psi_{l_1, \,l_2}$ together with their counterparts $\Gamma_K\, \psi_{l_1, \, l_2}$ form a basis of the space of all spinors on $T^2$. The kernel of $\sD^K$ is spanned by $\psi_{0, \, 0}$ and $\Gamma_K\, \psi_{0,\,  0}$. The eigenspinors of the Dirac operator on the flat torus, for all possible spin structures, are described in \cite{Ginoux, Friedrich}.

Each eigenspinor $\psi_{l_1, \, l_2}$ determines the properties of the 4D fermion $\varphi$ coupled to it in $\Psi = \varphi \otimes \psi_{l_1, \, l_2}$, which is a spinor over $M_4\times T^2$. For instance, the integer $l_j$ determines the charge of $\varphi$ with respect to the 4D gauge field linked to the Killing field $\partial_{x^j}$ on $K$. Here $(x^1, x^2)$ are the periodic Euclidean coordinates on the torus.

Now consider a 4D gauge field linked to an arbitrary vector field $X$ on $T^2$, not necessarily Killing. It couples to fermions through the derivative $\cL_X$. Decompose the $\sD^K\!$-eigenspinors as sums of Weyl components, $ \psi_{l_1, \, l_2} =  \psi_{l_1, \, l_2}^+ +  \psi_{l_1, \, l_2}^-$. Section \ref{ExplicitExampleT2} shows that
\beq
\cL_X\, \psi_{l_1, \, l_2}^\pm \; = \; 2\pi i\,  (X^1 l_1 + X^2 l_2) \,  \psi_{l_1, \, l_2}^\pm \; \mp \; \frac{i}{4} \, \divergence (JX)\, \psi_{l_1, \, l_2}^\pm \ ,
\eeq
 where $J$ denotes the complex structure on $T^2$. So when $\divergence (JX)$ is not zero, $\cL_X$ will act differently on $V^+$ and $V^-$ and will not have strong chiral symmetry. The Killing vector fields $J\, \partial_{x^j}$, however, do have vanishing divergence on $T^2$. So the $\cL_{\partial_{x^j}}$ and their linear combinations define the same $T^2$-representations on $V^+$ and $V^-$, as expected.

Finally, denote by $\cL_X^\pm$ the restriction of the Kosmann-Lichnerowicz derivative to $V^\pm$. Consider the simpler case where $X_h$ is a Hamiltonian vector field determined by a function $h$ on the torus. Section \ref{ExplicitExampleT2} also establishes that, for a generic $X_h$, there exists no invertible, unitary, $T^2$-equivariant linear map $\Theta: V^+ \rightarrow V^-$ that commutes with $\sD^2$ and satisfies $\Theta^{-1} \cL_{X_h}^- \Theta = \cL_{X_h}^+$. So $\cL_{X_h}$ does not have weak chiral symmetries.

\subsection*{Conserved currents and an additional comment}

Section \ref{ConservedCurrents} briefly investigates the conserved currents associated with solutions of the Dirac equation on $P$. Proposition \ref{SpinorCurrents} describes a natural relation between higher-dimensional and 4D currents. This follows from a general procedure to average higher-dimensional vector fields along the fibres of $P$, described in appendix \ref{SectionFibreProjection}. Section \ref{ConservedCurrents} also defines charge currents associated with Killing vector fields on $P$.

Now a final note. This paper investigates the properties of spinors satisfying the higher-dimensional Dirac equation $\sD^P \Psi = 0$, or the Weyl equation with constraints $\Gamma_P \Psi = \pm \Psi$. It studies how these spinors couple to the 4D objects that determine the background metric on $P$, such as the metric $g_M$, the gauge fields $A^a_\mu$, and the Higgs-like scalars coming from the components of $g_K$. These couplings determine how the higher-dimensional spinors are perceived in 4D. In this paper we are primarily concerned with the features that produce chiral interactions of fermions and 4D gauge fields. 
However, let us also mention a different point. A very well-known and elegant feature of the Kaluza-Klein framework is that the massless equation $\sD^P \Psi = 0$ produces non-zero mass terms in the 4D Dirac equation, after dimensional reduction. This also arises in section \ref{SectionDecompositionDiracOperator}, of course. Thus, if $\sD^P \Psi = 0$ or its Weyl counterpart were the physical equations of motion, this would mean that the fermionic masses observed in 4D are entirely due to the vibration of the higher-dimensional spinor along the internal space. In the heuristic picture provided by geodesics, this corresponds to the statement that the 4D rest energy of test particles is entirely due to their internal motion along $K$. As stressed in \cite{Bap2}, this is equivalent to saying that elementary particles always travel at the speed of light in higher dimensions. It is the projection of velocities to three dimensions that appears to produce speeds in the range $[0, c]$, as observed macroscopically. This simple and attractive picture is one of the conceptual rewards of working with a higher-dimensional spacetime.

\section{Spinors on Riemannian submersions}
\label{Submersions}

\subsection{Riemannian submersions}

This section describes properties of spinors on manifolds equipped with metrics defining Riemannian submersions. These metrics generalize the Kaluza ansatz by encoding not only the 4D metric and 4D massless gauge fields, but also massive gauge fields and an internal metric that can vary along $M_4$. The main classical results about Riemannian submersions were developed in \cite{ONeill1, Ehresmann, Hermann} and are presented in \cite{Besse, FIP}, for example. We use the translation of those results to the Kaluza-Klein language provided in \cite{Bap}. Spinors on Riemannian submersions have been previously studied in \cite{Moroianu} and more generally in \cite{Reynolds}. In this section, we introduce all the notation and reframe the geometric formalism using the explicit language suitable for Kaluza-Klein physics. It does not present new results.

Take a Lorentzian metric $g_P$ on the higher-dimensional space $P =  M_4 \times K$ such that the projection $P \rightarrow M_4$ is a Riemannian submersion. As described in \cite{Bap}, this is equivalent to taking a $g_P$ determined by three simpler objects: 
\begin{itemize}
\item[{\bf i)}]  a Lorentzian metric $g_M$ on the base $M_4$; 
\vspace{-.1cm}
\item[{\bf ii)}]  a family of Riemannian metrics $g_K(x)$ on the fibres $K_x$ parameterized by the points $ x \in M_4$;
\vspace{-.2cm}
\item[{\bf iii)}] a gauge one-form $A$ on $M_4$ with values in the Lie algebra of vector fields on $K$.
\end{itemize}
These objects determine the higher-dimensional metric through the relations
\bal \label{MetricDecomposition}
g_P (U, V) \ &= \ g_K (U, V) \nonumber \\
g_P (X, V) \ &= \  - \ g_K \left(A (X), V \right) \nonumber \\
g_P (X, Y) \ &= \ g_M (X, Y) \ + \  g_K \left(A(X) , A(Y) \right) \ ,
\end{align}
valid for all tangent vectors $X,Y \in TM$ and vertical vectors $U, V \in TK$. These relations generalize the usual Kaluza ansatz for $g_P$. Choosing a set $\{ e_a \}$ of independent vector fields on $K$, the one-form on spacetime can be decomposed as a sum 
\beq \label{GaugeFieldExpansion}
A(X) \ = \ \sum\nolimits_a \,A^a(X) \, e_a \ ,
\eeq
where the real-valued coefficients $A^a(X)$ are the traditional gauge fields on $M_4$. For general submersion metrics on $P$ this can be an infinite sum, with $\{ e_a \}$ being a basis of the full space of vector fields on $K$, which is the Lie algebra of the diffeomorphism group ${\rm Diff} (K)$. The gauge group need not act on $K$ only through isometries of $g_K$.

The curvature $F_A$ is a two-form on $M_4$ with values in the Lie algebra of vector fields on $K$. It can be defined by 
\beq
\label{CurvatureDefinition}
F_{A} (X, Y)  \ := \ (\dd_M A^a) (X, Y) \,\, e_a  \ + \ A^a (X)\, A^b (Y) \, [e_a, e_b]  \ ,
\eeq
where the last term is just the Lie bracket $ [A(X), A(Y)] $ of vector fields on $K$.

The tangent space to $P$ has a natural decomposition $TP = TM \oplus TK$. Then $TK$ is the kernel of the projection $TP \rightarrow TM$. It is also called the vertical sub-bundle $\VV$ of $TP$. Its $g_P$-orthogonal complement, denoted $\HH := (TK)^\perp$, is called the horizontal sub-bundle. So we get a second, $g_P$-dependent decomposition
\beq \label{HorizontalDistribution}
T P \; =\; \HH \oplus \VV   \ .
\eeq
Using these two decompositions, any tangent vector $w\in TP$ can be written as a sum $w= w_M + w_K = w^\HH  +  w^\VV$ in two different ways. The relation between them is 
\beq \label{DefinitionHorizontalDistribution}
w^\VV  \ = \   w_K \; - \;  A (w_M)    \qquad \qquad  w^\HH  \ = \  w_M \; + \;   A (w_M)   \ .
\eeq
So the information contained in the gauge one-form $A$ on $M$ is equivalent to the information contained in the horizontal distribution $\HH \subset TP$. For example, the curvature $F_A$ is the obstruction to the integrability of the distribution $\HH$ \cite{Besse}.

One can construct local, $g_P$-orthonormal trivializations of $TP$ using only horizontal and vertical vectors. They can take the form $\{ X_\mu^\HH, v_j \}$. Here the $v_j$ form an orthonormal basis of $TK$ with respect to $g_K (x)$, for each $x \in M$. The $X_\mu$ form a $g_M$-orthonormal basis of $TM$. The horizontal lift of $X_\mu$ to $P$ is denoted $X_\mu^\HH$, and is given by
\beq
\label{BasicLiftX}
X_\mu^\HH =  X_\mu \; + \; A^a(X_\mu) \; e_a  \ .
\eeq
Such horizontal lifts are called basic vector fields on $P$.

\subsection{Decomposing the covariant derivative of vector fields}

As before, let $\pi: P \rightarrow M$ be a Riemannian submersion with a metric $g_P$ equivalent to a triple $(g_M, A, g_K)$, as in \eqref{MetricDecomposition}. Denote by $\nabla$, $\nabla^M$ and $\nabla^K$ the Levi-Civita connections on $P$, on $M$, and on each fibre $K_x$. Let $U$, $V$ and $W$ be vertical vector fields on $P$ and let $X$, $Y$ and $Z$ be vector fields on the base $M$. 
Then we have the standard geometrical identities \cite[p. 240]{Besse}:
\begin{subequations}
\label{IdentitiesSubmersions}
\begin{align}
\label{IdentityCurvature1}
[X^\HH, \, Y^\HH]^\VV \; &= \; F^a_A(X, Y)\, \,  e_a    \linebr
\label{IdentityCurvature2}
g_P(\nabla_{X^\HH} Y^\HH, \,  U) \; &= \; \frac{1}{2} \,\, F_A^a (X, Y) \, g_P (e_a, U)   \linebr
\label{IdentitySubmersionH1}
g_P(\nabla_{X^\HH} Y^\HH, \, Z^\HH) \; &= \;  g_M (\nabla^M_{X} Y, \, Z) \   \linebr
\label{IdentitySecondFundamentalForm}
g_P(\nabla_U V, \, X^\HH) \;  &= \; -\, \frac{1}{2} \, (\dd^A g_K)_X (U, V)    \linebr
\label{IdentitySubmersionH2}
g_P(\nabla_U V, \, W) \; &= \;  g_K (\nabla^K_U  V , \, W)   
\end{align}
\end{subequations}
These are identities of vector fields and functions on $P$. Functions defined on $M$, namely $F^a_A(X, Y)$ and $ g_M (\nabla^M_{X} Y, Z)$, are regarded as functions on $P$ that are constant along the fibres.
The derivative $\dd^A g_K$ can be identified with the second fundamental form of the fibres of $P$. It measures how $g_K$ changes along $M$ up to diffeomorphisms of $K$. It is equivariant under $\Diff (K)$-gauge transformations.  As in \cite{Bap}, it can be expressed in terms of Lie derivatives and the gauge one-forms $A^a$ as 
\beq \label{DefinitionCovariantDerivative}
(\dd^A g_K)_X (U, V) \ = \  (\partial_{X} g_K) (U, V) \ + \ A^a (X) \; (\Lie_{e_a}\, g_K) (U, V) \ .
\eeq
Here $\partial_{X} g_K$ denotes the directional derivative of $g_K (x)$ along $X \in TM$ and $\Lie_{e_a}\, g_K$ denotes the Lie derivative along the internal vector field $e_a$.

\vspace{.2cm}

\begin{remarknotation} The notation used in this paper differs from the notation in the literature on Riemannian submersions. The modification is necessary to avoid a conflict with the traditional notation in physics. Namely, the tensor $A$ in \cite{ONeill1, Besse, FIP} is essentially what we call $F_A$ here, since it represents the physical gauge field strength.
The tensor $T$ in \cite{ONeill1, Besse, FIP} is called here $\dd^A g_K$. This avoids confusion with torsion and emphasizes its physical role as the covariant derivative of Higgs-like fields. The precise relations are
\beq \label{TranslationTensors2}
(\dd^A g_K)_X  (U, V)  \ = \  (\Lie_{X^\HH} g_P) (U, V) \ = \ -\, 2\,  g_P (T_U V, \, X^\HH) \ = \  -\, 2\, g_P (\nabla_U V, \, X^\HH  ) \ .
\eeq
The last equality is the definition of $T$ in \cite{ONeill1, Besse, FIP}. The first two equalities are derived in sections 2.3 and 2.5 of \cite{Bap}.
\end{remarknotation}

\subsection{Spinors on $M_4 \times K$}
\label{SpinorsOnP}

\subsubsection{Gamma matrices and chirality operators}
\label{GammaMatrices}
Let $k$ denote the dimension of the internal space $K$. The higher-dimensional spinor space $\Delta_{3+k,1}$ can be written as the tensor product $\Delta_{3,1} \otimes \Delta_k$. These spaces have irreducible representations of Clifford algebras. There is a standard isomorphism between $\Cl(3+k,1)$ and the $\mathbb{Z}_2$-graded tensor product of algebras $\Cl(3,1) \, \hat{\otimes}  \! \Cl(k)$. It is determined by the correspondence of generators
\begin{align}
\label{CliffordMultiplication}
1 &= 1\otimes 1   \nonumber   \linebr
\Gamma_\mu &= \gamma_\mu \otimes 1 \quad \ \ {\rm for} \ \ \mu = 0,1, 2, 3       \nonumber  \linebr
\Gamma_{3+j} &= \gamma_5 \otimes \rho_j \quad \ {\rm for} \ \ j = 1, ..., k
\end{align}
This is a recipe to construct higher-dimensional gamma matrices $\Gamma_l$ from the lower-dimensional ones. The $\rho_j$ are gamma matrices acting on $\Delta_k$, and the $ \gamma_\mu$ are 4D gamma matrices for the metric $\eta = \diag (-1, 1, 1, 1)$. The convention used here is that gamma matrices in spatial dimensions are square roots of $-1$ and are anti-self-adjoint with respect to the positive-definite inner-product of spinors $\langle  \Psi ,  \Psi \rangle =  \Psi ^\dag  \Psi $. This is the most common convention in Riemannian geometry \cite{LM, Bourguignon, Ginoux}. So for example
\bal
\label{ConventionsGammaMatrices}
\{ \gamma_\mu,\, \gamma_\nu \} \ = \ - 2 \, \eta_{\mu \nu} \, I_4  \qquad \quad (\gamma_0)^\dag \ &= \ \gamma_0  \qquad \quad (\gamma_l)^\dag \ = \ - \gamma_l  \ ,
\end{align}
for $l = 1, 2, 3$.  On a Lorentzian space, one can define an indefinite inner-product of spinors with respect to which all gamma matrices are self-adjoint operators. This is 
  \beq
  \label{Lorentzian pairing}
 \langle \Psi,  \Psi \rangle_{\Gamma_0} \ := \ \langle \Gamma_0 \, \Psi ,  \Psi \rangle \ .
  \eeq
The complex chirality operators on the spinor spaces $\Delta^\CC_{3,1}$,  $\Delta^\CC_{k}$ and $\Delta^\CC_{3+k,1}$ are defined by
\begin{align}
\label{ChiralOperators}
\gamma_5 \; &:= \; i\, \gamma_0 \,\gamma_1\, \gamma_2\, \gamma_3  \qquad \qquad \qquad  \qquad  \qquad  \quad   \Gamma_K\; := \; i^{[\frac{k+1}{2}]}\, \rho_1 \cdots \rho_k    \nonumber \linebr
\Gamma_P \; &:= \;  i^{[\frac{k+3}{2}]}\, \, \Gamma_0 \, \Gamma_1 \cdots \Gamma_{k+3}  \; = \;   (\gamma_5)^{k+1}  \otimes \Gamma_K  \ ,
\end{align}
where $[s]$ denotes the integral part of $s$. The choices of $i$-factors are very standard for $\gamma_5$ in the Minkowski case and for $\Gamma_K$ in the Riemannian case.
Using the Clifford relations, such as \eqref{ConventionsGammaMatrices}, one can check that these operators square as
\beq
\label{SquaresChiralOperators}
\gamma_5 \, \gamma_5 = 1 \qquad  \qquad \Gamma_K \,\Gamma_K = 1  \qquad  \qquad \Gamma_P \,\Gamma_P = 1
\eeq
on the respective spinor spaces. The three operators are self-adjoint with respect to the positive inner-product $\langle \cdot \, ,  \cdot \rangle$ of spinors. Moreover, $\gamma_5$ is anti-self-adjoint with respect to the $\langle \cdot ,  \cdot\rangle_{\gamma_0}$ pairing while $\Gamma_P$ satisfies 
\beq
\label{HermitianConjChiralOperator}
\langle\, \Phi ,  \Gamma_P \Psi \, \rangle_{\Gamma_0} \; =\;  (-1)^{k+1} \, \langle \, \Gamma_P \Phi , \Psi \, \rangle_{\Gamma_0} \ .
\eeq

\subsubsection{Horizontal and vertical spinors}

If $E$ is an oriented, real vector bundle with a spin structure, its complexified spinor bundle is denoted by $S_\CC(E)$. For a tangent bundle, the notation is simplified, as in $S_\CC(TP) = S_\CC(P)$. Now assume that $K$ has a spin structure. Together with the trivial spin structure on $M_4$, it determines a unique spin structure on $P = M_4 \times K$. We fix those structures for the rest of the paper. 

The horizontal and vertical bundles in decomposition \eqref{HorizontalDistribution} have their own associated spinor bundles, denoted $S_\CC(\HH)$ and $S_\CC(\VV)$.  Sections of $S_\CC(\HH)$ are called horizontal spinors, while sections of $S_\CC(\VV)$ are the vertical spinors. Calling $\pi$ the projection from $P$ to $M_4$, there is a natural isomorphism $\HH \simeq \pi^*(TM)$ as Lorentzian vector bundles over $P$. So the respective spinor bundles are isomorphic too, 
\beq 
\label{IsomorphismHorizontalBundle}
S_\CC(\HH) \; \simeq \;  S_\CC (\pi^\ast(TM)) \; \simeq \; \pi^\ast [S_\CC (M)] \ .
\eeq
In particular, a spinor $\varphi$ on $M_4$ has a unique lift as a horizontal spinor that coincides with the pull-back $\pi^\ast \varphi$ under that isomorphism. This is denoted $\varphi^\HH$ and called the basic lift of $\varphi$ to $P$. At the same time, if we fix a point $x \in M_4$ and consider the fibre $K_x := \{ x\} \times K$ inside $P$ with its metric $g_K (x)$, there is a natural isomorphism between the restriction of $S_\CC(\VV)$ to $K_x$ and the spinor bundle of the fibre,
\beq
\label{IsomorphismVerticalBundle}
S_\CC(\VV) \, |_{K_x} \; \simeq \; S_\CC(K_x)  \ .
\eeq
Now let $X$ be a vector field on $M_4$ and let $X^\HH$ denote its basic lift to $P$, as in \eqref{BasicLiftX}. Then Clifford multiplication in $S_\CC (M)$ and $S_\CC (\HH)$ satisfies the equivariance property
\beq
\label{EquivarianceClifford}
(X\cdot \varphi)^\HH \ = \ X^\HH \cdot \varphi^\HH
\eeq
for any spinor $\varphi$ on $M_4$. Moreover, a dissection of these constructions, as done in \cite{Reynolds}, shows that there is a natural isomorphism of spinor bundles
\beq
\label{TensorProductSpinorBundle}
S_\CC(P) \ \simeq \ S_\CC(\HH) \otimes S_\CC(\VV)  \ 
\eeq
that is compatible with the Clifford multiplication implied by \eqref{CliffordMultiplication}, in the sense that 
\begin{align}
\label{EquivarianceClifford2}
U \cdot (\varphi ^\HH \otimes \psi) \ &= \ (\gamma_5 \cdot \varphi) ^\HH \otimes (U \cdot \psi)  \linebr
X^\HH \cdot (\varphi ^\HH \otimes \psi) \ &= \ (X \cdot \varphi) ^\HH \otimes \psi \ . \nonumber
\end{align}
Here $U$ is any vertical vector field on $P$. It is regarded as such on the left-hand side and as a section of $\VV$ on the right-hand side. As before, $X$ is any vector in $TM$.

Since $M_4$ is contractible, its spinor bundle $S_\CC (M)$ is trivial. Due to \eqref{IsomorphismHorizontalBundle}, so is $S_\CC(\HH)$ as a bundle over $P$. This implies that a spinor $\Psi$ on $P$ can always be written as a sum 
 \beq
 \label{TensorHDSpinor}
\Psi (x,y) \; = \; \sum_{b=1}^4 \ \varphi_b^\HH(x) \otimes \psi^b(x,y) \ ,
\eeq
where $\varphi_b$ is a Dirac spinor on $M_4$ and $\psi^b$ is a vertical spinor on $P$. Here $x$ and $y$ denote coordinates on $M_4$ and $K$, respectively.

Regarding Dirac operators, let $\{X_\mu\}$ denote a $g_M$-orthonormal trivialization of $TM$ and let $\{v_j\}$ denote a local, $g_K$-orthonormal trivialization of the vertical bundle $\VV \rightarrow P$. Using the Levi-Civita connections of $g_M$ and $g_K$, one can define the Dirac operators 
\beq
\sD^M  \varphi  \;= \; i\,g_M^{\mu \nu} \, \, X_\mu \cdot \nabla^M_{X_\nu} \varphi  \qquad \qquad  \sD^K \psi \; =\;  g_M^{ij} \, \, v_i \cdot  \nabla^K_{v_j} \psi \ .
\eeq
The first acts on spinors on $M_4$; the second acts on vertical spinors over $P$. With the conventions \eqref{ConventionsGammaMatrices}, the first is self-adjoint with respect to the $L^2$-pairing $\blangle \gamma^0 \varphi, \, \varphi \brangle_{L^2}$ of spinors on $M_4$. The second is self-adjoint with respect to the product $\blangle \psi, \, \psi \brangle_{L^2}$ of vertical spinors, where the integration is taken over $K$ only.

When $K$ is compact, the vertical spinors over a fibre $\{x\} \times K$ can always be written as a (possibly infinite) sum of eigenspinors of the internal Dirac operator $\sD^K$. Since the metric $g_K$ depends on $x$, the operator $\sD^K$ and its eigenspinors will also change along $M_4$, in general. Now suppose that the internal metric $g_K (x)$ is independent of $x$. Then an $L^2$-orthonormal basis of eigenspinors $\{ \psi^\alpha (y)\}$ on $K$ applies transversally throughout $M_4$. In particular, we can take each $\psi^b$ in \eqref{TensorHDSpinor} and expand its $y$-dependence as $\psi^b(x, y) = \sum_\alpha c^b_{\alpha}(x) \, \psi^\alpha (y)$. Inserting this in \eqref{TensorHDSpinor}, it is clear that the higher-dimensional spinor can then be written as a (possibly infinite) sum
\beq
\label{EigenMassDecompositionSpinors}
\Psi (x,y) \ = \ \sum\nolimits_{\alpha} \ \varphi_\alpha^\HH(x) \otimes \psi^\alpha(y) \ .
\eeq
Here the $\varphi_\alpha$ are Dirac spinors on $M_4$, the $ \varphi_\alpha^\HH$ are their horizontal lifts to $P$, and the $\psi^\alpha$ are eigenspinors of $\sD^K$ independent of the point on $M_4$.

\section{Higher-dimensional Dirac operator}
\label{SectionDecompositionDiracOperator}

\subsection{Decomposing the covariant derivative of spinors}

The Levi-Civita connection on the tangent bundle $TP$ can be lifted to a connection on the spinor bundle. Denoting both connections by $\nabla$, the standard local formula is \cite{Bourguignon}:
\begin{align}
\label{StandardSpinorCovariantDerivative}
\nabla_X \Psi \; &:= \; \partial_X \tilde{\Psi} \;+\;  \frac{1}{4}\, g_P^{ir} \, g_P^{js} \, \, g_P(\nabla_X u_i, u_j) \, u_r \cdot u_s \cdot  \tilde{\Psi}    \linebr
&\,\, = \; \partial_X  \tilde{\Psi}  \;+ \; \frac{1}{8} \,  g_P^{ir} \, g_P^{js}\,  \big[ \, g_P(\nabla_X u_i, u_j)  -  g_P(\nabla_X u_j,  u_i) \, \big]  \,  u_r \cdot u_s \cdot  \tilde{\Psi}  \ . \nonumber
\end{align}
Here $\{u_r\}$ denotes a local, oriented, $g_P$-orthonormal trivialization of $TP$, while $\tilde{\Psi}$ represents the spinor $\Psi$ in the induced trivialization of $S_\CC (P)$. So $\tilde{\Psi}$ is a local function on $P$ with values in some $\CC^{r}$ and we denote by $\partial_X  \tilde{\Psi} = (\dd \tilde{\Psi})(X)$ its directional derivative. We are also using the fact that $g_P(\nabla_X u_i, u_j)$ is anti-symmetric in the two indices, for the Levi-Civita connection. Since the metric is Lorentzian, the elements $g_P^{rs}$ can be $\pm 1$. Below, we will most often abuse notation and simply write $\Psi$ instead of $\tilde{\Psi}$.

On a Lorentzian manifold, the lifted Levi-Civita connection is compatible with the $\langle \cdot ,  \cdot  \rangle_{\Gamma_0}$ inner-product of spinors defined in \eqref{Lorentzian pairing}, in the sense that
\beq
\Lie_X \,  \langle \Psi ,  \Psi \rangle_{\Gamma_0} \; = \;  \langle  \nabla_X \Psi \, , \Psi \rangle_{\Gamma_0}  \; + \;  \langle \Psi \, , \nabla_X \Psi \rangle_{\Gamma_0}
\eeq
for any vector field $X$. On a Riemannian manifold, the lifted Levi-Civita connection is compatible with the positive-definite inner-product $\langle \cdot ,  \cdot  \rangle$.

As described in section \ref{Submersions}, for a Riemannian submersion $g_P \simeq (g_M, A, g_K)$, we can take the orthonormal trivialization of $TP$ to be of the form $\{ X_\mu^\HH, v_j \}$, so a collection of basic and vertical vector fields on $P$. In this adapted trivialization, the covariant derivatives are given by the following formulae, proved in appendix \ref{DecompositionHDOperators}.
\begin{proposition}  
\label{thm:DecompositionCovariantDerivatives}
Consider a spinor on $P$ of the form $\Psi = \varphi^\HH(x) \otimes \psi(x,y)$, as in \eqref{TensorHDSpinor}. Its covariant derivative along a vertical vector field $U$ is given by
\begin{align}
\label{VerticalCovariantDerivative}
\nabla_{U} \Psi \; =& \; \varphi^\HH \otimes(\nabla_{U} \psi)   \; - \;   \frac{1}{8} \, \, g_M^{\mu \nu}\, g_M^{\sigma \rho}\, \, (F_A^a)_{\mu \sigma } \, \, g_K(e_a , U) \; (X_\nu \cdot X_\rho \cdot \varphi)^\HH \otimes \psi   \nonumber \linebr
&+ \, \frac{1}{4} \,  \, g_M^{\mu \nu}\,  g_K^{i j}\, \, (\dd^A g_K)_{X_\mu} (U, v_i)\; (  X_\nu \cdot \gamma_5 \cdot \varphi  )^\HH \otimes (v_j\cdot \psi) \ .
\end{align}
Its covariant derivative along a horizontal, basic vector field $Y^\HH$ on $P$ is given by
\begin{align}
\label{HorizontalCovariantDerivative}
\nabla_{Y^\HH} \Psi \; =& \;  (\nabla^M_Y \varphi)^\HH \otimes \psi \; + \;  A^a(Y) \,  \, \varphi^\HH \otimes  \cL_{e_a} \psi  \linebr 
\begin{split}
 &+ \, \frac{1}{4}\, \, g_M^{\mu \nu}\, \, F_A^a(Y , X_\mu) \, \, (  X_\nu \cdot \gamma_5 \cdot \varphi )^\HH \otimes (e_a\cdot \psi)   \linebr
 &+\;  \varphi^\HH \otimes \Big( \partial_Y \psi  \, + \,  \frac{1}{8} \, \, g_K^{i r}\,g_K^{j s}  \, \big\{ g_P(\, [Y, v_i], v_j)  -  g_P(\, [Y, v_j ] , v_i) \big\} \, v_r \cdot v_s \cdot \psi  \Big) \ . \nonumber
\end{split}
\end{align}
Here $ \cL_{e_a} \psi$ denotes the Kosmann-Lichnerowicz derivative of spinors on the fibres $K$.
\end{proposition}
These expressions are more explicit and expanded versions of those given in \cite{Reynolds} for the case of simple submersions. Importantly, they reveal that the 4D gauge fields $A^a$ couple to spinors through the Kosmann-Lichnerowicz derivative $\cL_{e_a}$ written in \eqref{DefinitionKLDerivative}.

These expressions are simpler in regions where the Higgs-like scalars are constant, i.e. where the internal metric $g_K$ does not vary along $M_4$. There, the orthonormal vector fields $\{v_j\}$ on $K$ can be taken to be independent of the coordinate $ x\in M_4$. So $[Y, v_i] = 0$, since $Y$ is a vector field on $M_4$. If the vertical spinor is also taken to be independent of $x$, as in \eqref{EigenMassDecompositionSpinors}, the last line of \eqref{HorizontalCovariantDerivative} vanishes entirely and we get:
\begin{corollary}
In regions where the internal metric $g_K$ and the vertical spinor $\psi$ are independent of the coordinate $x \in M_4$, the previous expressions reduce to
\begin{align}
\label{HorizontalCovariantDerivative2A}
\nabla_{U} \Psi \, =& \; \varphi^\HH \otimes(\nabla_{U} \psi)   \; - \;   \frac{1}{8} \, \, g_M^{\mu \nu}\, g_M^{\sigma \rho}\, \, (F_A^a)_{\mu \sigma } \, \, g_K(e_a , U) \; (X_\nu \cdot X_\rho \cdot \varphi)^\HH \otimes \psi    \nonumber \linebr
&+ \frac{1}{4}  \, g_M^{\mu \nu} \, g_K^{i j}\, \, A^a(X_\mu) \, \, (\Lie_{e_a} g_K)(U, v_i)\; (  X_\nu \cdot \gamma_5 \cdot \varphi )^\HH \otimes (v_j\cdot \psi)  \linebr
\label{HorizontalCovariantDerivative2B}
\nabla_{Y^\HH} \Psi  \, =& \;  (\nabla^M_Y \varphi)^\HH \otimes \psi \; + \; A^a(Y) \, \, \varphi^\HH \otimes \cL_{e_a} \psi   \nonumber \linebr 
\begin{split}
 &+ \frac{1}{4}\, g_M^{\mu \nu} \, \, F_A^a(Y , X_\mu) \, \, (  X_\nu \cdot \gamma_5 \cdot \varphi )^\HH \otimes (e_a\cdot \psi)  \ . 
\end{split}
\end{align}
\end{corollary}
\noindent
Here $\Lie_{e_a} g_K$ denotes the Lie derivative of $g_K$ along the internal vector field $e_a$. This term comes from the simplification of $\dd^A g_K$ in \eqref{VerticalCovariantDerivative}, using definition \eqref{DefinitionCovariantDerivative} and the fact that $g_K$ is constant along the flow of $X_\mu$ in $M_4$.

\subsection{Decomposing the higher-dimensional Dirac operator}

The decomposition of the spinor connection given in proposition \ref{thm:DecompositionCovariantDerivatives} leads directly to a decomposition of the Dirac operator on $P$. For a higher-dimensional spinor of the form $\varphi^\HH \otimes \psi$, this is described in the next result, proved in appendix \ref{DecompositionHDOperators}. 
\begin{proposition}
\label{thm:DecompositionDiracOperator}
Consider a spinor on $P$ of the form $\Psi = \varphi^\HH(x) \otimes \psi(x,y)$, as in \eqref{TensorHDSpinor}. The action of the higher-dimensional Dirac operator on $\Psi$ can be locally decomposed as 
\begin{align}
\label{GeneralFormulaDecompositionDiracOperator}
\sD^P \Psi \; =& \; \, g_M^{\mu \nu}\,  (X_\mu \cdot  \nabla^M_{\nu} \varphi)^\HH \otimes \psi  \; + \;   g_M^{\mu \nu}\, \, A_\nu^a\,  \, (X_\mu \cdot  \varphi)^\HH  \otimes  \big[ \cL_{e_a} + \frac{1}{2} \, \divergence (e_a) \big] \psi  \nonumber   \linebr
&+  \; ( \gamma_5 \cdot \varphi)^\HH  \otimes  \sD^K\psi   \; + \;  \frac{1}{8} \, \,  (F_A^a)^{\mu \nu } \,  (  X_\mu \cdot X_\nu \cdot \gamma_5 \cdot \varphi  )^\HH \otimes (e_a \cdot \psi)   \linebr
 &+\;  g_M^{\mu \nu}\,\, (X_\mu \cdot \varphi)^\HH \otimes \, \Big[  \partial_{X_\nu}  + \frac{1}{2}  \,\big( \partial_{X_\nu}   \log \sqrt{|g_K |}  \, \big)  \Big]  \, \psi  \  \nonumber \linebr
  &+\;  g_M^{\mu \nu}\,  (X_\mu \cdot \varphi)^\HH \otimes \Big( \frac{1}{8} \, \sum\nolimits_{ij}   \big\{ g_P([X_\nu, v_i], v_j)  -  g_P([X_\nu, v_j ] , v_i) \big\}  \, v_i \cdot v_j \cdot \psi  \Big) \, . \nonumber 
 \end{align}
 Here $\cL_{e_a}$ denotes the derivative \eqref{DefinitionKLDerivative} of spinors on $K$; $\divergence (e_a)$ denotes the divergence of the internal vector field $e_a$ with respect to $g_K$; and $|g_K |$ is the modulus of the determinant of the matrix representing $g_K$ in a fixed coordinate system on $K$. 
\end{proposition}
This expression is a more explicit and expanded version of that given in \cite{Reynolds}. The detail is necessary to better understand the underlying physics.
Formula \eqref{GeneralFormulaDecompositionDiracOperator} simplifies in regions where the Higgs-like scalars are constant, i.e. where the internal metric $g_K$ does not change along $M_4$. In this case $[X, v_i] = 0$, as before, and every higher-dimensional spinor can be expressed as a sum of simpler tensor products of the form $\varphi(x)^\HH \otimes \psi(y)$, as in \eqref{EigenMassDecompositionSpinors}. Thus, in this simpler setting, the Dirac operator acts as follows.
\begin{corollary}
In regions where $g_K$ is constant along $M_4$, the action of the Dirac operator on a spinor of the form $\varphi(x)^\HH \otimes \psi(y)$ can be decomposed as 
\begin{align}
\label{SimplerDecompositionDiracOperator}
\sD^P (\varphi^\HH \otimes \psi ) \; =& \; \, g_M^{\mu \nu}\,  (X_\mu \cdot  \nabla^M_{\nu} \varphi)^\HH \otimes \psi  \; + \;   g_M^{\mu \nu}\, \, A_\nu^a\,  \, (X_\mu \cdot  \varphi)^\HH  \otimes  \big[ \cL_{e_a} + \frac{1}{2} \, \divergence (e_a) \big] \psi  \nonumber   \linebr
&+  \; ( \gamma_5 \cdot \varphi)^\HH  \otimes  \sD^K\psi   \; + \;  \frac{1}{8} \, \,  (F_A^a)^{\mu \nu } \,  (  X_\mu \cdot X_\nu \cdot \gamma_5 \cdot \varphi  )^\HH \otimes (e_a \cdot \psi) \ .
 \end{align}
\end{corollary}
This expression is valid for a general gauge one-form $A_\mu = A^a_\mu \, e_a$ on $M_4$ with values in the space of vector fields on $K$, be they Killing or non-Killing with respect to $g_K$. The first term in \eqref{SimplerDecompositionDiracOperator} contains the Dirac operator in 4D. The second term determines the couplings between gauge fields and fermions. The term with the internal Dirac operator generates mass terms for fermions in 4D. The last term is a Pauli-type coupling between the gauge field strength and spinors. It is a standard feature in Kaluza-Klein dimensional reductions.

Now suppose that a higher-dimensional spinor of the form $\Psi (x,y)  =   \varphi_\alpha^\HH(x) \otimes \psi^\alpha(y)$, as in \eqref{EigenMassDecompositionSpinors}, satisfies the massless Dirac equation $\sD^P \Psi = 0$.  The $\{ \psi^\alpha \}$ form an $L^2$-orthonormal basis of $\sD^K\!$-eigenspinors on $K$. In particular, they are independent of each other. Thus, writing $m_\alpha$ for the respective eigenvalue and using the traditional $\gamma^\mu$ notation, equation \eqref{SimplerDecompositionDiracOperator} implies that each $\varphi_\alpha$ must satisfy the equation over $M_4$:
\beq
\label{DimensionReductionDiracEquation}
 \gamma^\mu \, \nabla^{M, \, A}_{X_\mu} \,  \varphi_\alpha \; +\; m_\alpha \, \gamma_5 \cdot \varphi_\alpha \; +\; \frac{1}{8}\, \,  (F_A^a)_{\mu \nu } \, \blangle \psi_\alpha \, , e_a \cdot \psi_\beta \brangle_{L^2}\, \,  \gamma^\mu  \gamma^\nu  \gamma_5 \, \varphi_\beta    \;= \; 0  \ .
\eeq  
Here we are summing over the index $\beta$ and the 4D covariant derivative is defined by 
\beq
\label{Induced4DCovariantDerivative}
\nabla^{M, \, A}_{X_\mu} \,  \varphi_\alpha \; := \; \nabla^{M}_{X_\mu} \,  \varphi_\alpha \; + \; A^a_\mu \,  \,  \Blangle \, \psi_\alpha  \, , \big[ \cL_{e_a} + \frac{1}{2} \, \divergence (e_a) \big] \psi_\beta \, \Brangle_{L^2}\,  \, \varphi_\beta   \ .
\eeq
The reader may notice that \eqref{DimensionReductionDiracEquation} is similar to, but not identical to, the 4D Dirac equation. Besides the Pauli term, it has extra $\gamma_5$ factors and a non-standard kinetic term. However, denoting 
\beq
\varphi_\alpha' \; := \;  \frac{1}{\sqrt{2}} \, ( I    +  i  \gamma_5 ) \, \varphi_\alpha \; = \; \exp{(i \pi \gamma_5 / 4)}  \, \varphi_\alpha  \ ,
\eeq
as in \cite[p. 22]{Duff}, it is easy to check that the redefined 4D spinors satisfy 
\beq
\label{DimensionReductionDiracEquation2}
 i\, \gamma^\mu \, \nabla^{M, \, A}_{X_\mu} \,  \varphi_\alpha' \; +\; m_\alpha \, \varphi_\alpha' \; +\; \frac{1}{8}\, \,  (F_A^a)_{\mu \nu } \, \blangle \psi_\alpha \, , e_a \cdot \psi_\beta \brangle_{L^2}\, \,  \gamma^\mu  \gamma^\nu \, \varphi_\beta'    \;= \; 0  \ .
\eeq 
This is the traditional Dirac equation, with a Pauli term, written in Minkowski space with signature $-+++$ and the gamma matrix conventions described in \ref{GammaMatrices}.

Thus, the physical couplings of gauge fields to 4D fermions, as given by the covariant derivative \eqref{Induced4DCovariantDerivative}, are determined by the matrix elements of $\cL_{e_a}$ in the basis $\{ \psi_\alpha \}$ of eigenspinors of $\sD^K$. These matrix elements are anti-hermitian. The identity
\beq
 \Blangle \,  \big[ \cL_{e_a} + \frac{1}{2} \, \divergence (e_a) \big] \psi_\alpha \, ,  \psi_\beta \, \Brangle_{L^2}  \; + \;  \Blangle \, \psi_\alpha  \, , \big[ \cL_{e_a} + \frac{1}{2} \, \divergence (e_a) \big] \psi_\beta \, \Brangle_{L^2} \; = \;  0 \ 
\eeq
follows from the general property \eqref{IntegralKLD} of the Kosmann-Lichnerowicz derivative. 

\begin{remark}
The appearance of a 4D Pauli term is a departure from the prescriptions of the Standard Model. In the case of the abelian 5D Kaluza-Klein model, this term contributes to the anomalous magnetic moment of the charged fermion. It has been estimated that the magnitude of this contribution is negligible, beyond the reach of current measurements (e.g. \cite{CN, Dolan}). However, those estimates rely on the assumption that the higher-dimensional vacuum metric is determined by the simple Einstein-Hilbert action on $P$. That assumption should not hold in realistic models. Higher-order corrections to that action seem necessary for several reasons. For example, to introduce the different mass scales observed in reality and to stabilize the vacuum metric on $P$ \cite{Bap}.
\end{remark}

\begin{remark}
Consider the operator on spinors over $P$ defined by the Pauli term,
\beq
C (\varphi^\HH \otimes \psi) \, := \,   \frac{1}{8} \, \,  (F_A^a)^{\mu \nu } \,  (  X_\mu \cdot X_\nu \cdot \gamma_5 \cdot \varphi  )^\HH \otimes (e_a \cdot \psi) \ .
\eeq
It is algebraic and anti-self-adjoint with respect to the pairing $\langle \cdot, \cdot \rangle_{\Gamma_0}$ of spinors. Thus, the modified operator $\sD^P\! - C$ retains most of the useful properties of $\sD^P$, such as ellipticity, anti-self-adjointness with respect to $\langle \cdot, \cdot \rangle_{\Gamma_0}$, and the implicit coupling of 4D gauge fields to spinors through the Kosmann-Lichnerowicz derivative. So one could also consider $(\sD^P\! - C) \Psi = 0$ as a candidate for the physical equation of motion for spinors on $P$. An important conceptual disadvantage of the modified operator, however, is that it is defined only for submersion metrics of the form $g_P \simeq (g_M, A, g_K)$, not for general metrics on $P$.
\end{remark}

\section{Properties of the Kosmann-Lichnerowicz derivative}
\label{PropertiesKLD}

\subsection{General properties}
\label{GPropertiesKLD}

This section collects several useful properties of the Kosmann-Lichnerowicz derivative of spinors, denoted $\cL_X$. This operator was introduced by Lichnerowicz for a Killing $X$ in \cite{Lich}. Later, Kosmann extended it to any $X$ and thoroughly investigated its properties \cite{Kosmann}.  Most results presented here can be found in the literature in the Riemannian case \cite{Kosmann, BT}, but are somewhat dispersed and are stated using various conventions. The simple formulae in lemma \ref{thm: KLDerivativeRS} are necessary for sections \ref{ExplicitExampleT2} and \ref{ExplicitExampleS2} but do not seem to exist in the literature. So their proofs are given here. All formulae written in the present section \ref{GPropertiesKLD} remain valid in the semi-Riemannian case.

Let $(K, g)$ be a semi-Riemannian manifold with a fixed spin structure and spinor bundle $S_\CC(K) \rightarrow K$. Given a vector field $X$ and a spinor $\psi$ on $K$, the Kosmann-Lichnerowicz derivative is defined by
\begin{align}
\label{DefinitionKLDerivative}
\cL_X \, \psi \; &:=\;  \nabla_X \psi \; -\; \frac{1}{8}\, g^{ir} g^{js} \big[  g(\nabla_{v_r} X, v_s) - g(\nabla_{v_s} X, v_s) \big] \, v_i \cdot v_j  \cdot \psi     \linebr
&\, \, = \;  \partial_X \tilde{\psi} \;+ \;\frac{1}{8}\,  g^{ir} g^{js} \big\{  g([X, v_r], v_s) - g([X, v_s], v_r) \big\} \,   v_i \cdot v_j  \cdot \tilde{\psi} \ . \nonumber
\end{align}
Here $\nabla$ denotes both the Levi-Civita connection on $TK$ and its standard lift \eqref{StandardSpinorCovariantDerivative} to the spinor bundle. The dot denotes the Clifford action of vectors on spinors.  The vectors $v_j$ form a $g$-orthonormal, oriented, local trivialization of $TK$. The second line is a local expression using the representative $\tilde{\psi}$ of $\psi$  with respect to the trivialization of $S_\CC(K)$ induced by $\{ v_j\}$. It follows from definition \eqref{StandardSpinorCovariantDerivative} and a direct application of the Koszul formula for $\nabla$. When $X$ is Killing, the coefficients $g_K(\nabla_{v_r} X, v_s)$ and $g_K([X, v_r], v_s)$ are both anti-symmetric in $r$ and $s$. 

The operator $\cL_X$  is $\mathbb{C}$-linear and acts as a derivation when the spinor is multiplied by a function in $C^\infty (K)$,
\beq
\cL_X \, (\psi_1 \, + \, f\,  \psi_2) \ = \ \cL_X  \psi_1 \; + \; f\, \cL_X \psi_2 \; + \; (\Lie_X f) \, \psi_2 \ .
\eeq
Its commutation relation with the Clifford multiplication of vectors and spinors is 
\beq
\label{CommutatorCliffordMultiplicationKLD}
\cL_X (Y \cdot \psi) \; - \; Y \cdot \cL_X \psi  \; = \; [X,Y] \cdot \psi \; +\; \frac{1}{2}\, g^{ij}  \, (\Lie_X g)(v_i, Y)\; v_j \cdot \psi \ .
\eeq
Unlike the covariant derivative, it does not behave nicely when the vector field $X$ is multiplied by a scalar function:
\beq
\cL_{f X} \, \psi \ = \  f\, \cL_X \psi \; + \;   \frac{1}{8} \, \big[ (\gradient f) \cdot X - X \cdot (\gradient f) \big] \cdot \psi \ .
\eeq
Using that the volume element is covariantly constant on $K$, one can derive that $\cL_X$ always commutes with the chirality operator:
\beq
\label{CommutatorChiralOperatorKLD}
\cL_{X} \, \Gamma_K \,  \psi \; = \;   \Gamma_K  \, \cL_{X} \,  \psi \ .
\eeq
The commutator of $\cL_X$ with the covariant derivative $\nabla$ is encapsulated in
\begin{align}
\label{CommutatorConnectionKLD}
(\cL_X \nabla)_Y \, \psi  \; &:= \; \cL_X (\nabla_Y \, \psi ) \; - \;  \nabla_Y( \cL_X \, \psi ) \; - \; \nabla_{[X, Y]}\, \psi      \linebr 
&\,=\; \frac{1}{8}\, g^{ir} g^{js}   \big\{ [\nabla_{v_r}(\Lie_X g)] (v_s , Y)  \; - \;    [\nabla_{v_s}(\Lie_X g)] (v_r , Y)  \big\} \; v_i \cdot v_j \cdot \psi \nonumber
\end{align}
The operator $\cL_X \nabla$  is also called the derivative of the connection $\nabla$. The symbol $\Lie_X g$ denotes the Lie derivative of the metric $g$ along $X$. It is again a section of $T^\ast K \otimes T^\ast K$, like the metric itself, and has a covariant derivative $\nabla (\Lie_X g)$. 

The commutator of $\cL_X$ with the standard Dirac operator is given by
\begin{align}
\label{CommutatorDiracKLD}
[\sD, \cL_X] \, \psi  \ =& \ \frac{1}{2}\, g^{ir} g^{js} \, (\Lie_X g) (v_r , v_s) \ v_i \cdot \nabla_{v_j} \psi  \  \linebr
 &\, + \frac{1}{4}\, g^{ir} g^{js}  \big\{ [\nabla_{v_r}(\Lie_X g)] (v_i , v_j)  \; - \;    [\nabla_{v_j}(\Lie_X g)] (v_i , v_r)  \big\} \; v_s \cdot \psi  \  \nonumber
\end{align}
Just like the Levi-Civita connection \eqref{StandardSpinorCovariantDerivative}, the Kosmann-Lichnerowicz derivative is compatible with the inner-product of spinors. This means that
\beq
\label{CompatibilityKLDInnerProduct}
 ( \cL_X \psi_1 , \psi_2 )\, +\, ( \psi_1 , \cL_X \psi_2 ) \; = \; \Lie_X ( \psi_1 , \psi_2 ) \ ,
\eeq
where $\Lie_X$ denotes the standard Lie derivative of functions on $K$.
 Using Stokes' theorem, this implies that for compactly supported spinors we have 
\beq
\label{IntegralKLD}
\int_K \big\{ ( \cL_X \psi_1  , \psi_2 )\, +\, ( \psi_1 , \cL_X \psi_2 ) \, + \,  \divergence(X)\, ( \psi_1, \psi_2 ) \big\} \, \vol_{g} \ = \ 0 \ ,
\eeq
where $ \divergence(X)$ denotes the divergence of $X$ on $(K, g)$. So $\cL_X$ is formally anti-self-adjoint for a divergence-free $X$.
\begin{remark}
\label{PairingRemark}
 In formulae \eqref{CompatibilityKLDInnerProduct} and \eqref{IntegralKLD}, the inner-product of spinors should be interpreted with care. If the metric $g$ has a signature with $r$ minus signs, then $( \cdot , \cdot )$ denotes the inner-product such that 
\beq
\label{AdjointnessCliffordPairing}
(X \cdot \psi_1 , \psi_2 ) \; = \; (-1)^{r-1} (\psi_1 ,  X \cdot \psi_2 )
\eeq
for any $X \in TK$ \cite[p. 68]{Baum}. Thus, it is the positive-definite pairing $\langle \psi, \psi \rangle = \psi^\dag \psi$ in the Riemannian case and the indefinite pairing $\langle \psi, \psi \rangle_{\Gamma_0} = \langle \Gamma_0 \psi, \psi \rangle$ in the Lorentzian case (see the conventions in section \ref{SpinorsOnP}).
\end{remark}
The Kosmann-Lichnerowicz derivative is not a proper Lie derivative of spinors, since it does not satisfy the usual closure relation. In general we have
\begin{align}
\label{NonClosureKLDerivative}
( \, [\cL_X, \cL_Y] \, -\, \cL_{[X, Y]} \,  ) \, \psi  \ =& \  \frac{1}{4}\, g^{ir} g^{js}  g^{kl} \, \big\{ \, (\Lie_X g)(v_r , v_k) \;  (\Lie_Y g)(v_s , v_l)  \linebr 
 &\, - \;  (\Lie_X g)(v_s , v_k) \;  (\Lie_Y g)(v_r, v_l)  \, \big\} \; v_i \cdot v_j \cdot \psi  \ . \nonumber
\end{align}
However, the right-hand side does vanish when either $X$ or $Y$ are Killing or conformal Killing vector fields on $K$. When this is not the case, the right-hand side only vanishes to first order in the Lie derivatives of $g$.

\begin{remark}
Despite being a higher-order correction, the right-hand side of \eqref{NonClosureKLDerivative} represents a significant conceptual variation from current practice in the Standard Model and grand unified theories \cite{Hamilton, BH}. It shows how there are natural equations of motion, in this case the equation $\sD \Psi = 0$ of a Kaluza-Klein model, in which only the unbroken part of the 4D gauge group couples to fermions through an exact Lie algebra representation. The broken part does not, and understandably so. It does not preserve the internal vacuum metric, and, since spinors are objects that depend on the metric, it does not induce a standard symmetry transformation of spinors. As in the case of 4D general relativity, it would be relevant to better understand how to transform spinors and the Dirac operator under general diffeomorphisms of the underlying space.
\end{remark}

\subsection{Kosmann-Lichnerowicz derivative on Riemann surfaces}
\label{KLDRiemannSurfaces}

This section derives two formulae for $\cL_X$ that will be needed in the examples in sections \ref{ExplicitExampleT2} and \ref{ExplicitExampleS2}. Let $K$ be a Riemann surface with metric $g$ and compatible complex structure $J$. As with general K\"ahler manifolds \cite[ch. 2]{Besse}, we have the compatibility relations
\begin{align} \label{KahlerRelations}
g(JX, JY) \; &=\; g(X, Y)  \qquad \qquad  \quad  \vol_g (X, Y)\; = \; g(JX, Y) \  \nonumber  \linebr
\nabla_X (JY) \; &= \; J \, \nabla_X Y
\end{align}
for any vector fields $X$ and $Y$ on $K$. Choosing isothermal coordinates on $K$ \cite[sec. 1.5]{IT}, the metric can be locally written as 
\beq
g \; = \;  \Omega^2 (x,y) \, [ \, (\dd x)^2 \, + \, (\dd y)^2 \,] 
\eeq
for some positive function $\Omega$. The complex structure acts on the coordinate vector fields as $J(\partial_x) = \partial_y$ and $J(\partial_y) = - \partial_x$,
so $z = x+iy$ is a complex coordinate on $K$. Thus, the local vector fields $v_1 = \Omega^{-1} \partial_x$ and $v_2 =  \Omega^{-1}  \partial_y$ form an oriented, $g$-orthonormal trivialization of $TK$ satisfying $J v_1 = v_2$.

Now fix a spin structure and spinor bundle $S_\CC(K)$. This is a rank-two complex vector bundle over $K$. Definition \eqref{ChiralOperators} says that the chirality operator acts on spinors as
\beq
\label{ChiralOperatorRS}
\Gamma_K  \, \psi  \; = \;  i \, v_1 \cdot v_2 \cdot \psi \ .
\eeq
Then we have:
\begin{lemma}
\label{thm: KLDerivativeRS}
The Kosmann-Lichnerowicz derivative of a spinor on a Riemann surface can be written as 
\beq
\label{KLDerivativeRS}
\cL_X \psi \ = \  \nabla_X \psi \; - \;  \frac{i}{4} \,\divergence(JX) \, \Gamma_K \, \psi 
\eeq
for any vector field $X$ on $K$. In particular, when $X_h$ is a Hamiltonian vector field determined by a function $h$, we have
\beq
\label{KLDerivativeRSH}
\cL_{X_h} \psi \ = \  \nabla_{X_h} \psi \; - \;  \frac{i}{4} \,(\Delta h) \, \Gamma_K \, \psi \ ,
\eeq 
where $\Delta h$ denotes the Laplacian of $h$. 
\end{lemma}

\begin{proof}
Using the compatibility relations \eqref{KahlerRelations}, the divergence of the vector field $JX$ is
\begin{align}
\label{DivergenceJX}
\divergence (JX) \; &=  \; \sum\nolimits_j g( \nabla_{v_j} JX , v_j) \; = \; \sum\nolimits_j g( J \nabla_{v_j} X , v_j) \; = \; - \sum\nolimits_j g( \nabla_{v_j} X , J v_j)  \nonumber \linebr 
&= \;  g( \nabla_{v_2} X , v_1 ) \, - \,  g( \nabla_{v_1} X , v_2 )  \ .
\end{align}
Then \eqref{KLDerivativeRS} follows by combining \eqref{DivergenceJX}, \eqref{ChiralOperatorRS} and the definition \eqref{DefinitionKLDerivative} of $\cL_X$.
Now let $h$ be a real, smooth function on $K$ and let $X_h$ denote its Hamiltonian vector field \cite{Cannas}. It is a divergence-free vector field determined by the equality of 1-forms
\beq
\label{DefinitionHamiltonianField}
\dd h \; = \; \iota_{X_h} \vol_g \ .
\eeq
It follows from \eqref{KahlerRelations} that
\beq
\label{IdentityHamiltonianFieldRS}
J X_h \; = \; \gradient \, h \ 
\eeq
as vector fields on $K$. So $\divergence (J X_h) = \Delta h$ and \eqref{KLDerivativeRSH} follows from \eqref{KLDerivativeRS}. 
\end{proof}

Hamiltonian vector fields always have vanishing divergence. When $K = S^2$, every divergence-free vector field is Hamiltonian for an appropriately chosen function $h$. This holds because the first de Rham cohomology group of $S^2$ vanishes.

\section{Chiral fermions}
\label{ChiralFermions}

\subsection{Correlation between internal and 4D chiralities}
\label{CorrelationChiralities}

This section starts by reviewing well-known observations about the correlation between internal and 4D chiralities (e.g. \cite{Witten83}) in our setting of general Riemannian submersions. This is a mechanism to explain how the chirality of interactions between 4D fermions and  gauge fields can be caused, in principle, by an asymmetry of the action of the Kosmann-Lichnerowicz derivative $\cL_X$ on left- and right-handed internal spinors. Section \ref{InternalChiralFermions} then studies the chiral properties of $\cL_X$ in more detail. It defines the concepts of chiral isomorphism and strong and weak chiral symmetries. 

 In this section, we assume that $K$ is a compact, connected, even-dimensional, spin manifold with a fixed topological spin structure. For any metric $g_K$, the chirality operator satisfies $\Gamma_K \Gamma_K = 1$. It has eigenvalues $\pm 1$ and its eigenspaces determine a decomposition of the complex spinor bundle as a sum of half-spinor bundles, $S_\CC(K)  =  S^+_\CC(K)  \oplus  S^-_\CC(K)$. Since this is true for any metric, the isomorphisms \eqref{IsomorphismVerticalBundle} on the fibres determine a splitting of the vertical bundle on $P$ as a sum $S_\CC^+(\VV) \oplus S_\CC^-(\VV)$. Similarly, the decomposition $S_\CC(M_4)  =  S^+_\CC(M_4)  \oplus  S^-_\CC(M_4)$ on the base together with the isomorphism \eqref{IsomorphismHorizontalBundle} determines a splitting of the horizontal bundle as $S_\CC^+(\HH) \oplus S_\CC^-(\HH)$. Using the formula $\Gamma_P = \gamma_5 \otimes \Gamma_K$, which comes from \eqref{ChiralOperators} for an even-dimensional $K$, it is clear that these decompositions are related to the higher-dimensional splitting $S_\CC(P)  =  S^+_\CC(P)  \oplus  S^-_\CC(P)$ simply by
\begin{align}
\label{TensorProductSpinorBundle2}
S^+_\CC(P) \ &\simeq \ \big[ S^+_\CC(\HH) \otimes S^+_\CC(\VV) \big] \oplus  \big[ S^-_\CC(\HH) \otimes S^-_\CC(\VV) \big] \linebr
S^-_\CC(P) \ &\simeq \ \big[ S^+_\CC(\HH) \otimes S^-_\CC(\VV) \big] \oplus  \big[ S^-_\CC(\HH) \otimes S^+_\CC(\VV) \big] \ . \nonumber
\end{align}
Sections of $S^\pm_\CC(P)$ are the Weyl spinors on $P$. Now suppose that the higher-dimensional physical spinors are taken to be sections of $S^+_\CC(P)$, say, instead on sections of $S_\CC(P)$. In other words, suppose that physical spinors are taken to satisfy the right-handed Weyl equation: 
\beq
\sD^P \, \Psi \; = \; 0 \qquad \qquad \Gamma_P \, \Psi \; = \; \Psi \ .
\eeq
Then it follows from \eqref{TensorProductSpinorBundle2} that the solutions will be sums of spinors of the form
\beq
\Psi (x,y) \; = \; \varphi^+(x) \otimes \psi^+(x,y) \; + \; \varphi^-(x) \otimes \psi^- (x,y) \ .
\eeq
So the right-handed 4D fermions are coupled only to right-handed internal spinors, and similarly for their left-handed counterparts. The higher-dimensional Weyl equation imposes a correlation between 4D and internal chiralities. Now suppose that a 4D gauge field is linked to a vector field $X$ on $K$ and the derivative $\cL_X$ acts differently on the spaces of sections of $S^+_\CC(\VV)$ and $S^-_\CC(\VV)$. Then, it follows from \eqref{Induced4DCovariantDerivative} that the 4D chiral components $\varphi^\pm$ will have different couplings to that gauge field. So the emergence of chiral fermions in 4D could be due to an asymmetry of the action of $\cL_X$ on the chiral components $\psi^\pm$ of vertical spinors.

Given the isomorphisms \eqref{IsomorphismVerticalBundle}, our task is therefore to study the symmetries of $\cL_X$ when acting on the spaces of sections of $S^\pm_\CC(K)$, denoted $V^\pm$. That is the aim of section \ref{InternalChiralFermions}. For now, let us introduce notation and recall standard facts that will be necessary there. For a compact $K$, the Dirac operator $\sD_{g_\KKK}$ is elliptic and self-adjoint with respect to the $L^2$-inner-product of complex spinors 
\beq
\label{L2InnerProductSpinors}
\blangle \,\psi , \,\psi' \, \brangle_{L^2} \ := \ \int_{K} \langle \,\psi , \,\psi' \, \rangle \, \vol_{g_\KKK} \ .
\eeq 
Here $\langle \,\psi , \,\psi' \, \rangle$ denotes the positive-definite, hermitian pairing that is $i$-linear in the second entry and $i$-antilinear in the first one. The Dirac operator has finite-dimensional eigenspaces, denoted $V_m$, with eigenvalues $m$ that are real, discrete and accumulate only at infinity. The space of $L^2$-integrable spinors then decomposes as an orthogonal Hilbert direct sum of the different $V_m$. 

For a compact $K$, the isometry group of $g_K$ is a compact Lie group, and we denote by $G$ its connected component. We always assume that the $G$-action on $K$ can be lifted to an action on the spin structure. When $G$ is non-trivial, the index of the Dirac operator vanishes for a connected, even-dimensional $K$ \cite{AH}. This means that the kernel of $\sD_{g_\KKK}$ splits into a sum of two $\Gamma_K$-eigenspaces of equal dimension, $V_0 = V_0^ + \oplus V_0^-$. When $G$ is trivial, the index of $\sD_{g_\KKK}$ can still vanish. This will happen when $K$ admits any metric with positive scalar curvature, for example. When $g_K$ itself has this property, the Schr\"odinger-Lichnerowicz formula guarantees that the kernel of $\sD_{g_\KKK}$ is trivial.

\subsection{Internal chiral symmetries}
\label{InternalChiralFermions}

Assume that $K$ is a compact, connected, even-dimensional, spin manifold with a fixed metric $g_K$ and spin structure. Its spinor bundle decomposes as a direct sum of half-spinor bundles, $S^\pm_\CC(K)$, whose space of sections are denoted $V^\pm$. The Kosmann-Lichnerowicz derivative $\cL_X$ always commutes with the chirality operator $\Gamma_K$. So it preserves the spaces $V^\pm$.  We want to study whether $\cL_X$ acts similarly on $V^+$ and $V^-$. When this happens, we say that $\cL_X$ has chiral symmetry. This notion should imply an identity of matrix elements $\langle \phi^+, \cL_X \psi^+ \rangle = \langle \phi^-, \cL_X \psi^- \rangle$ for all $\sD$-eigenspinors $\phi$ and $\psi$ with positive eigenvalues. In fact, we will distinguish two notions of chiral symmetry: when that identity holds exactly (strong symmetry), and when it holds only up to a unitary redefinition of the $\sD$-eigenspinors (weak symmetry). In the latter case, the redefinitions must respect the symmetries of spinors induced by the isometry group of $g_K$.

Any spinor on $K$ can be written as a sum of its chiral components $\psi = \psi^+ + \psi^-$, as in \eqref{DefChiralProjections}.
Let $V_m$ denote the eigenspace of $\sD$ with real eigenvalue $m$. Since $\Gamma_K$ anti-commutes with $\sD$, it preserves $V_0$ and takes $V_m$ isomorphically to $V_{-m}$ when $m\neq 0$. In particular, it preserves the sums $V_m \oplus V_{-m}$ for all $m>0$. These are eigenspaces of $\sD^2$. Since $\Gamma_K$ squares to the identity, each of these invariant spaces can be further decomposed into eigenspaces of $\Gamma_K$ with eigenvalues $\pm 1$. We write
\begin{align}
\label{DecompositionsSpinorSpaces3}
V_0 \; &= \; V_0^+ \oplus V_{0}^- \linebr
V_m \oplus V_{-m} \; &= \; V_m^+ \oplus V_{m}^-  \ , \qquad \quad m>0 \ . \nonumber
\end{align}
Let $G$ denote the connected component of the isometry group of $g_K$. Its action on spinors preserves the eigenspaces $V_m^\pm$. This follows because $\cL_Y$ commutes with $\Gamma_K$ and $\sD$ when $Y$ is Killing. Thus, for a non-trivial $G$, we can further decompose each $V_m^\pm$ as a sum of irreducible representation $G$-spaces. The summands are labelled by the respective $\sD$-eigenvalue $m$, the representation $\rho$ and a non-negative integer $n_{m, \rho}$ accounting for the representation multiplicity. So we write\footnote{For the trivial representation, $V_{m, \rho}^\pm$ is unidimensional and $n_{m, \rho}$ is the dimension of the subspace of $V_m^\pm$ where $G$ acts trivially.}
\beq
\label{DecompositionsSpinorSpaces}
V^\pm \; = \;  \bigoplus_{m\geq0}\, V_m^\pm \; = \; \bigoplus_{m \geq 0} \bigoplus_{\rho} \,n_{m, \rho} \, V_{m, \rho}^\pm \ .
\eeq
The integers $n_{m, \rho}$ are the same for both decompositions $\pm$. This is because, for all $m>0$, there is a natural, unitary, $G$-equivariant isomorphism 
\beq
\label{NaturalUnitaryIsomorphism}
I_m: V_m^+ \longrightarrow V_m^-  \qquad \quad \psi \longmapsto m^{-1} \sD \psi \ .
\eeq
Each $I_m$ is just a scalar multiple of the Dirac operator inside $V_m^+$. Its $G$-equivariance follows from the fact that $\sD$ commutes with the $G$-action on spinors. Since $\sD$ anti-commutes with $\Gamma_K$, the definitions \eqref{DefChiralProjections} of chiral projections imply that 
\beq
\label{NaturalUnitaryIsomorphism2}
I_m (\psi^+) \; = \; \psi^-
\eeq
 for all $\psi$ in the $\sD$-eigenspaces $V_m$ with eigenvalue $m>0$.
In the case $m=0$, the map $I_m$ is not well defined. So the difference $n^+_{0, \rho} \, - \, n^-_{0, \rho}$ could be non-zero. This difference is called the $\rho$-index of $\sD$. However, a result of Atiyah and Hirzebruch guarantees that it does vanish for a compact $K$ with non-trivial $G$. Thus, in this case, even for $m=0$, the subspaces $V_0^+$ and $V_0^-$ have the same decomposition into irreducible $G$-spaces \cite{AH, Witten83}.

\vspace{.2cm}

\begin{definition}
\label{ChiralIsomorphism}
A chirality isomorphism is an invertible, unitary, $G$-equivariant, linear map $\Theta: V^+ \rightarrow V^-$ that commutes with $\sD^2$.
\end{definition}
Such isomorphisms exist if and only if the index of $\sD$ vanishes. If $V_0^+$ and $V_0^-$ do not have the same dimension, such invertible maps cannot be constructed. In the other direction, if $V_0^+$ and $V_0^-$ have the same dimension, there exists a non-canonical unitary isomorphism $\Theta_0$ between them. If $G$ is non-trivial, the aforementioned result of Atiyah and Hirzebruch guarantees that $\Theta_0$ can be chosen to be $G$-equivariant. Then a full isomorphism $\Theta$ can be defined by making it coincide with the map $I_m$ of \eqref{NaturalUnitaryIsomorphism} over $V_m^+$, for each $m>0$, and coinciding with $\Theta_0$ over $V_0^+$.

Now let $X$ be a vector field on $K$ and let $\cL_X^\pm$ denote the restriction of the Kosmann-Lichnerowicz derivative to the spaces $V^\pm$ of Weyl spinors. We define:
\begin{definition}
\label{DefWeakChiralSymmetry}
The derivative $\cL_X$ is said to have weak chiral symmetry if there exists a chirality isomorphism $\Theta $ such that $\Theta^{-1} \cL_X^- \, \Theta = \cL_X^+$ as operators on $V^+$.
\end{definition}
\noindent
Since chirality isomorphisms are unitary, this is the same as saying that
\beq
\label{WCSit2}
\blangle \, \phi,\, \cL_X \psi\, \brangle_{L^2}  \; = \; \blangle\, \Theta \, \phi, \, \cL_X \Theta \,\psi  \,\brangle_{L^2}
\eeq
 for all spinors $\phi$ and $\psi$ in $V^+$. The particular $\Theta$ that fulfils this condition may depend on $X$, of course. If it can be chosen to coincide with the map $I_m$ for all indices $m>0$ while preserving property \eqref{WCSit2}, we will say that $\cL_X$ has strong chiral symmetry. In this case, for any $\psi$ in $V_m$ and $\phi$ in $V_{m'}$ with $m, m'>0$, it follows from \eqref{NaturalUnitaryIsomorphism2} and \eqref{WCSit2} that
\beq
\label{TautologicalIdentities1}
\blangle \phi^+, \cL_X \psi^+ \brangle_{L^2} \; = \; \blangle I_{m'} \phi^+, \cL_X  I_m \psi^+ \brangle_{L^2} \; = \; \blangle \phi^-, \cL_X \psi^- \brangle_{L^2} \ .
\eeq
In contrast, for spinors in $V_0$ the natural map $I_m$ of \eqref{NaturalUnitaryIsomorphism} is not defined. So in that subspace we must keep using the non-canonical isomorphism $\Theta$. In the definition of strong chiral symmetry we also demand that $\cL_X$ preserves $V_0$. In general, the Kosmann-Lichnerowicz derivative $\cL_X$ preserves the eigenspaces $V^\pm$ of $\Gamma_K$, but a priori need not preserve the smaller subspaces $V_m^\pm$ and $V_{m, \rho}^\pm$ when $X$ is not Killing.

\vspace{.2cm}

\begin{definition}
\label{DefStrongChiralSymmetry}
The derivative $\cL_X$ is said to have strong chiral symmetry if 
\begin{enumerate}
\item\label{SCSit1} $\blangle \phi^+, \cL_X \psi^+ \brangle_{L^2} = \blangle \phi^-, \cL_X \psi^- \brangle_{L^2}$ for all spinors $\phi$ and $\psi$ in $V_{m>0}$;
\item\label{SCSit2} It preserves the kernel of $\sD$;
\item\label{SCSit3} There exists a unitary, $G$-equivariant, linear map $\Theta_0: V_0^+ \rightarrow V_0^-$ such that \\
 $\blangle  \phi,  \, \cL_X \psi  \brangle_{L^2} =  \blangle  \Theta_0\, \phi,  \,\cL_X \Theta_0 \,\psi  \brangle_{L^2}$
 for all spinors $\phi$ and $\psi$ in $V_0^+$.
\end{enumerate}
\end{definition}
\begin{remark}
It is clear that a derivative $\cL_X$ satisfying the conditions of strong chiral symmetry will also satisfy those of weak symmetry. Keeping identities \eqref{TautologicalIdentities1} in mind, the map $\Theta$ of definition \ref{DefWeakChiralSymmetry} can be constructed to coincide with $I_m$ over each $V_m^+$ with $m>0$ and coinciding with the map $\Theta_0$ of definition \ref{DefStrongChiralSymmetry} over $V_0^+$.
\end{remark}

The traditional no-go arguments against the existence of chiral fermions in Kaluza-Klein models essentially amount to the statement that, when $X$ is Killing with respect to $g_K$, the derivative $\cL_X$ has strong chiral symmetry. This is the content of proposition \ref{thm: KillingCase}. This symmetry, together with the correlation of internal and 4D chiralities, forces left- and right-handed 4D fermions to have the same couplings to massless gauge fields.

The main message of the present paper, in contrast, is that the symmetries observed in the Killing case are not generic. They will not hold for most non-Killing vector fields on $K$, even if they are small perturbations of Killing fields. The explicit calculations in section \ref{ExplicitExampleS2} show that, for a generic Hamiltonian vector field on $K=S^2$, the derivative $\cL_{X}$ does not have strong chiral symmetry. This is the content of proposition \ref{thm: SphereCounterexample}.

The calculations in section \ref{ExplicitExampleT2} go further. They consider the case where $X_h$ is a Hamiltonian vector field on the torus $K=T^2$ equipped with its trivial spin structure. They show that, for a generic $X_h$, the derivative $\cL_{X_h}$ does not have strong nor weak chiral symmetries. This is the content of proposition \ref{thm: TorusCounterexample}. The absence of chiral symmetries should hold much more generally, beyond the explicit examples illustrated in this paper.

We end this section with a result that helps understand points \ref{SCSit1} and \ref{SCSit2} in the definition of strong chiral symmetry. Consider the $L^2$-orthogonal decompositions of the space of sections of $S_\CC (K)$:
\beq
\label{DecompositionsSpinorSpaces2}
V \; = \; V_{m>0} \oplus V_0 \oplus V_{m<0} \; = \;    V_{m>0}^+ \oplus  V_{m>0}^- \oplus V_0^+ \oplus V_0^- \ .
\eeq
Here we have defined $V_{m>0}^\pm := \bigoplus_{m > 0}\, V_m^\pm$ and the second equality uses decompositions \eqref{DecompositionsSpinorSpaces3}. Then we have a series of equivalences, proved in appendix \ref{ProofsChiralFermions}: 
\begin{proposition}
\label{EquivalentConditions}
Let $X$ be a divergence-free vector field on a compact, even-dimensional, spin manifold $(K, g_K)$. The following conditions are equivalent:
\begin{enumerate}
\item\label{it1} $\cL_X$ preserves the subspace of spinors $V_{m>0}$;
\item\label{it2} $\cL_X$ preserves all the subspaces appearing in decompositions \eqref{DecompositionsSpinorSpaces2};
\item\label{it3} the commutator $[\cL_X, \sD ]$ preserves the subspace $V_{m>0}$;
\item\label{it4} $\langle \phi^+\!, \cL_X \psi^+ \brangle_{L^2} = \langle \phi^- \!, \cL_X \psi^- \brangle_{L^2}$ for all $\psi$ in $V_{m>0}$ and all $\phi$ in $V_{m \geq 0}$.
\end{enumerate}
Moreover, points 1 and 2 in the definition of strong chiral symmetry can be replaced by any one of these conditions.
\end{proposition}

\section{Explicit example: \texorpdfstring{$K = T^2$}{K=T2}}
\label{ExplicitExampleT2}

This section describes the simplest explicit example of the ideas discussed so far. It considers the case where $K$ is the two-torus $T^2 = \mathbb{R}^2 / \mathbb{Z}^2$ equipped with its flat metric of total volume $1$ and trivial spin structure. The basic facts about spinors on flat tori are very well-known, even for the non-trivial spin structures. See, for example, references \cite{Friedrich} and \cite[ch. 2.1]{Ginoux}. After introducing the notation and the basic facts, the main results are presented in formulae \eqref{KLDT2Chiral}, \eqref{MatrixElementsKLDT2_Ham} and in proposition \ref{thm: AbsenceWeakSymmetryT2}.

The spinor bundle $S_\CC (T^2)$ determined by the trivial spin structure is a trivial, rank two complex vector bundle over $T^2$. Consider the family of sections characterized by the integers $(l_1, l_2) \in  \mathbb{Z}^2$,
\beq
\label{DefEigenspinorsT2}
\psi_{l_1, \, l_2} \ = \ \frac{1}{\sqrt{2}} \, \,  e^{2\pi i (l_1\, x_1 + l_2 \, x_2)} \, \bmatr  \,  1 \linebr \, c_{l_1, \, l_2} \  \ematr \ .
\eeq
Here $(x^1, x^2)$ are the periodic Euclidean coordinates on $T^2$ with the constant $c_{l_1,\, l_2}$ defined as 1 if $l_1=l_2=0$ and as the phase $(l_1 - i\, l_2) \,/ \, | l_1 - i\, l_2|$ otherwise. These spinors and their counterparts $\Gamma_K \, \psi_{l, \, n}$ are all eigenspinors of the Dirac operator, 
\beq
\sD \psi_{l_1, \, l_2} \; = \; 2\, \pi \,  \sqrt{l_1^2+l_2^2} \, \, \,  \psi_{l_1, \, l_2}   \qquad \qquad     \sD  ( \Gamma_K \,\psi_{l_1, \, l_2} ) \; = \;  -\, 2\, \pi  \,  \sqrt{l_1^2+l_2^2} \, \, \, \Gamma_K \, \psi_{l_1, \, l_2} \ .
\eeq
These are all the eigenspinors of the Dirac operator, up to linear combinations. So the $\psi_{l_1, \, l_2}$ and $\Gamma_K  \, \psi_{l_1, \, l_2}$ form a basis of the space of smooth spinors on $T^2$. This basis is orthonormal with respect to the $L^2$-inner-product \eqref{L2InnerProductSpinors}.

Now denote by $\partial_{x^j}$ the Killing vector field on $T^2$ generated by the coordinate $x^j$. Any vector field on $T^2$ can be decomposed as $X= X^1 \partial_{x^1} + X^2 \partial_{x^2}$. The Levi-Civita connection of the flat metric on $T^2$ is trivial, in the sense that its Christoffel symbols $\nabla_{\partial_{x^j}} \partial_{x^i}$ vanish. So it follows from \eqref{KLDerivativeRS} and the definition of $ \psi_{l_1, \, l_2}$ that the Kosmann-Lichnerowicz derivatives are given by
\beq
\label{KLD_T2}
\cL_X\, \psi_{l_1, \, l_2} \; = \; 2\pi i\,  (X^1 l_1 + X^2 l_2) \,  \psi_{l_1, \, l_2} \; - \; \frac{i}{4} \, \divergence (JX)\,\Gamma_K\,  \psi_{l_1, \, l_2} \ .
\eeq
For the flat metric and natural complex structure on $T^2 = \mathbb{R}^2 / \mathbb{Z}^2$, the divergence factor is simply $\divergence (JX) = \partial_{x^2} X^1 - \partial_{x^1} X^2$. In particular, 
\beq
\cL_{\partial_{x^j}} \, ( \psi_{l_1, \, l_2}) \; = \; 2\pi\, i\, l_j\, \, \psi_{l_1, \, l_2}    \ .
\eeq
Returning to the Kaluza-Klein model on $M_4 \times T^2$, this means that the integer $l_j$ will be interpreted in four dimensions as the $\partial_{x^j}$-charge of the 4D fermion $\varphi$ appearing in the tensor product $\varphi \otimes \psi_{l_1, \, l_2}$. In other words, it represents the charge of $\varphi$ when responding to the massless gauge field on $M_4$ linked to the internal Killing field $\partial_{x^j}$.

Now decompose the $\sD$-eigenspinors as sums of their chiral components, $ \psi_{l_1, \, l_2} =  \psi_{l_1, \, l_2}^+ +  \psi_{l_1, \, l_2}^-$. Formula \eqref{KLD_T2} implies that
\beq
\label{KLDT2Chiral}
\cL_X\, \psi_{l_1, \, l_2}^\pm \; = \; 2\pi i\,  (X^1 l_1 + X^2 l_2) \,  \psi_{l_1, \, l_2}^\pm \; \mp \; \frac{i}{4} \, \divergence (JX)\, \psi_{l_1, \, l_2}^\pm \ .
\eeq
So $\cL_X$ preserves the subspaces of spinors $V^\pm$, as it should. It is also clear that $\cL_X$ will act differently on $V^+$ and $V^-$ when $\divergence (JX)$ is not zero. Thus, in general, $\cL_X$ does not have strong chiral symmetry. However, the Killing fields $\partial_{x^j}$ do have vanishing divergence on $T^2$. So the $\cL_{\partial_{x^j}}$ and their linear combinations have strong symmetry, as expected.

To investigate the possible existence of weak chiral symmetry, start by observing that, in our torus setting, all the spaces $V_{m, \rho}^\pm$ in decompositions \eqref{DecompositionsSpinorSpaces} are one-dimensional and have multiplicity $n_{m, \rho} =1$. In particular, the restrictions of unitary chirality isomorphisms $\Theta : V^+_{l_1,  \, l_2} \rightarrow V^-_{l_1,  \, l_2}$ can differ only by complex phases. Thus, according to definition \eqref{WCSit2} of weak chiral symmetry, the question is whether we can redefine the Weyl spinors $\psi_{l_1, \, l_2}^-$ through those phases in order to obtain a perfect match 
\beq
\label{PerfectMatchT2}
\blangle \psi_{l_1, \, l_2}^+, \, \cL_{X_h} \psi_{n_1, \, n_2}^+ \brangle_{L^2}\;  \stackrel{?}{=}  \; \blangle \psi_{l_1, \, l_2}^-, \, \cL_{X_h} \psi_{n_1, \, n_2}^- \brangle_{L^2}
\eeq
for all integers $l_j$ and $n_j$. 

Consider the simpler case where $X_h$ is a Hamiltonian vector field determined by a function $h$ on the torus. Its components are $X_h^1= \partial_{x^2} h$ and $X_h^2= - \partial_{x^1} h$. Then a straightforward computation using Stokes' theorem yields:
\begin{proposition}
\label{thm: MatrixElementsKLDT2_Ham}
For a Hamiltonian vector field $X_h$ on $T^2$ and any integers $l_j$ and $n_j$,
\begin{align}
\label{MatrixElementsKLDT2_Ham}
\blangle \psi_{l_1, \, l_2}^+, \, \cL_{X_h} \psi_{n_1, \, n_2}^+ \brangle_{L^2} \; &= \;  i\;  \overbar{r_{l_1, l_2, n_1, n_2}}  \; h_{(l_1-n_1, l_2-n_2)}   \linebr
\blangle \psi_{l_1, \, l_2}^-, \, \cL_{X_h} \psi_{n_1, \, n_2}^- \brangle_{L^2} \; &= \; -\, i\;  \overbar{c_{l_1,l_2}} \, \, c_{n_1,n_2} \;  r_{l_1, l_2, n_1, n_2}  \; h_{(l_1-n_1, l_2-n_2)}  \ .  \nonumber
\end{align}
\end{proposition}
Here we have denoted the Fourier coefficients of a function on $T^2$ as
\beq
h_{(l_1, l_2)} \; := \; \int_{T^2} h\, \,  e^{-2\pi i (l_1\, x_1 + l_2 \, x_2)} \, \vol_{g_\KKK} \ .
\eeq
The complex phases $c_{l_1,l_2}$ are as in \eqref{DefEigenspinorsT2} and we have defined the complex numbers
\beq
r_{l_1, l_2, n_1, n_2} \; :=\;  \frac{1}{2} \, \pi^2 \, \big[ \, (l_1-n_1)^2 + (l_2-n_2)^2 + 4 i \, (l_1 n_2  - l_2 n_1)  \, \big] \ . 
\eeq
The matrix elements \eqref{MatrixElementsKLDT2_Ham} vanish when $l_j = n_j$. Since the $c_{l_1,l_2}$ are phases, it is clear that the $\pm$ matrix elements of $\cL_{X_h}$ always have the same absolute value for a Hamiltonian $X_h$. Nonetheless, it is not possible to redefine the Weyl spinors $\psi_{l_1, \, l_2}^-$ through complex phases to obtain the match \eqref{PerfectMatchT2}. To recognize this, write
\beq
\blangle \, \psi_{l_1, \, l_2}^+, \, \cL_{X_h} \psi_{n_1, \, n_2}^+ \, \brangle_{L^2} \; = \;   a_{l_1, l_2, n_1, n_2} \;  \blangle \,\overbar{c_{l_1,l_2}}\, \psi_{l_1, \, l_2}^-, \, \cL_{X_h}\, \overbar{c_{n_1,n_2}}\,  \psi_{n_1, \, n_2}^- \, \brangle_{L^2} \ 
\eeq
with coefficients
\beq
a_{l_1, l_2, n_1, n_2}  \; = \;   -\,  \frac{\overbar{r_{l_1, l_2, n_1, n_2}}}{r_{l_1, l_2, n_1, n_2}} \ .
\eeq
defined when $(l_1, l_2) \neq (n_1, n_2)$. Although these coefficients are complex phases, they are not separable and cannot be written as a product $a_{l_1, l_2, n_1, n_2}  \stackrel{?}{=} \overbar{a_{l_1, l_2}} \, a_{n_1, n_2}$ for a fixed  function $a_{l,n}$ on $\mathbb{Z}^2$ with values in $U(1)$. This is because they do not satisfy the identity
\beq
a_{l_1, l_2, n_1, n_2} \ a_{l'_1, l'_2, n'_1, n'_2} \; \stackrel{?}{=} \;  a_{l_1, l_2, n'_1, n'_2} \ a_{l'_1, l'_2, n_1, n_2} 
\eeq
for all values of the integer labels, as can be easily checked. Thus, the desired redefinition of the $\psi_{l_1, \, l_2}^-$ is not possible, and we conclude that:
\begin{proposition}
\label{thm: AbsenceWeakSymmetryT2}
Let $X_h$ be a generic Hamiltonian vector field on the torus $K=T^2$ equipped with its trivial spin structure. Then the Kosmann-Lichnerowicz derivative $\cL_{X_h}$ does not have strong or weak chiral symmetries.
\end{proposition}
Returning to the Kaluza-Klein model on $M_4 \times T^2$, this means that, generically, the massive 4D gauge field linked to $X_h$ has unavoidable chiral interactions with the 4D fermion $\varphi$ appearing in the tensor product $\varphi \otimes \psi_{n_1,\, n_2}$ for $(n_1, n_2) \neq (0,0)$. This is the 4D fermion with abelian charges $(n_1, n_2)$. The amount of chirality depends on $X_h$ through the Fourier components of the function $h$ on the torus.

\section{Explicit example: \texorpdfstring{$K = S^2$}{K=S2}}
\label{ExplicitExampleS2}

\subsection{Spinors on $S^2$}
\label{SpinorsS2}

This section provides a second explicit example of chiral interactions. For the vacuum internal space we take the two-sphere $S^2 \simeq \CC \mathbb{P}^1$ equipped with its round metric of radius 1, positive curvature and total volume $4\pi$. Part \ref{SpinorsS2} reviews standard facts about spinors on $S^2$. See references \cite{Bar} and \cite[ch. 9.A]{GVF}, for example. After the notation and basic facts are set, section \ref{ChiralFermionsS2} presents the calculations showing the emergence of chiral fermions.

Spheres have unique spin structures and spinor bundles. The complexified spinor bundle $S_\CC(S^2)$ is a trivial, rank-two complex vector bundle. It can be written as a direct sum of two non-trivial complex line bundles,
\beq \label{SpinorBundleDecomposition}
S_\CC (S^2)  \ = \ S_\CC^+ \,\oplus \, S_\CC^-  \ . 
\eeq
Then $S_\CC^+$ coincides with the hyperplane line bundle and $S_\CC^-$ with its dual, the tautological bundle. The holomorphic tangent and cotangent bundles of $S^2$ can be written in terms of these two line bundles as $TS^2 \simeq S_\CC^+ \otimes S_\CC^+$ and $T^*S^2 \simeq S_\CC^- \otimes S_\CC^-$. 

The infinite-dimensional space of $L^2$-integrable sections of $S_\CC^\pm$ has a standard orthonormal basis, the spin-weighted spherical harmonics $Y^{\pm \frac{1}{2}}_{l,\, n}$ \cite{NP, GMNRS}. We denote them simply by $Y^{\pm}_{l,\, n}$. The index $l$ runs through the positive half-integers $\{ \frac{1}{2}, \frac{3}{2}, \ldots \}$. For each value of $l$, the index $n$ runs through $\{-l, -l+1, \ldots, l\}$. Note that the $Y^{\pm}_{l,\, n}$ are not scalar functions on $S^2$; they are sections of non-trivial bundles. So their value at a point depends on the coordinates being used.

These spherical harmonics can be used to define sections of the spinor bundle. Using decomposition \eqref{SpinorBundleDecomposition}, define
\beq
\psi_{l, \, n} \ := \ \frac{1}{\sqrt{2}}  \bmatr  \ Y^{+}_{l,\, n} \linebr \ i\, Y^{-}_{l,\, n} \  \ematr \ .
\eeq
These spinors and their counterparts $\Gamma_K  \psi_{l, \, n}$ are all eigenspinors of the Dirac operator, 
\beq
\sD \psi_{l, \, n} \; = \; \Big( l+\frac{1}{2} \Big)\, \psi_{l, \, n}   \qquad \qquad   \quad  \sD  \big( \Gamma_K \, \psi_{l, \, n}\big) \; = \;  -\, \Big( l+\frac{1}{2} \Big)\, \Gamma_K \, \psi_{l, \, n} \ .
\eeq
Let  $(\theta, \phi)$ be the usual spherical coordinates on $S^2$ and denote by $\partial_\phi$ the Killing vector field generated by the azimuthal angle.  One can directly calculate the Kosmann-Lichnerowicz derivatives
\beq
\cL_{\partial_\phi} (\psi_{l, \, n}) \; = \; i\, n\, \psi_{l, \, n}   \qquad \qquad   \quad     \cL_{\partial_\phi} (\Gamma_K \, \psi_{l, \, n}) \; = \; i\, n\, \Gamma_K \,\psi_{l, \, n}             \ .
\eeq
So the spinors $\psi_{l, \, n}$ over $S^2$ are entirely characterized by their eigenvalues with respect to the commuting operators $\sD$ and $\cL_{\partial_\phi}$. For a Kaluza-Klein model on $M_4 \times S^2$, the eigenvalues of these internal operators determine the mass and the $\partial_\phi$-charge of a 4D fermion $\varphi$ appearing in the higher-dimensional spinor $\varphi \otimes \psi_{l, \, n}$.

Since the spin-weighted spherical harmonics form a basis of the space of sections of $S_\CC^\pm$, it follows from \eqref{SpinorBundleDecomposition} that the $\psi_{l, \, n}$ and $\Gamma_K  \psi_{l, \, n}$ form a basis of the space of $L^2$-integrable spinors on $S^2$. This basis is orthonormal with respect to the inner-product \eqref{L2InnerProductSpinors}.
On the sphere, the eigenspinors of $\sD$ with the smallest eigenvalues $\pm 1$ are actually Killing spinors. For any vector field $X$ on $S^2$, they satisfy
\beq
\label{KillingSpinorsS2}
\nabla_X \, \psi_{\frac{1}{2}, \, \pm \frac{1}{2}} \; = \; -\, \frac{1}{2} \, X \cdot \psi_{\frac{1}{2}, \, \pm \frac{1}{2}}  \qquad \qquad  \nabla_X \, \big(\Gamma_K \, \psi_{\frac{1}{2}, \, \pm \frac{1}{2}} \big) \; = \;  \frac{1}{2} \, X \cdot  \Gamma_K \, \psi_{\frac{1}{2}, \, \pm \frac{1}{2}} \ .
\eeq

Identifying the punctured sphere $S^2\setminus \{N\}$ with the complex plane through stereographic projection, the normalized Killing spinors can be written in coordinates as
\beq
\label{ExplicitKillingSpinors}
\psi_{\frac{1}{2}, \, \frac{1}{2}} (z, \bz) \ = \ \frac{1}{2\sqrt{\pi \, q}}  \bmatr  \; 1 \linebr \; i\, z \;  \ematr    \qquad \qquad    \psi_{\frac{1}{2}, \, -\frac{1}{2}} (z, \bz) \ = \ \frac{1}{2\sqrt{\pi \, q}}  \bmatr  \; \bz \linebr  \; - \, i \;  \ematr   \ ,
\eeq
where we have defined the function $q := 1 + |z|^2$.

\subsection{Chiral fermions with $K = S^2$}
\label{ChiralFermionsS2}

The aim of this section is to state and prove proposition \ref{thm: ChiralMatrixElementsS2}. It illustrates the emergence of chiral interactions of spinors on $S^2$.
Start by considering the Killing spinors $\psi_{\frac{1}{2}, \, \pm \frac{1}{2}}$ written in \eqref{ExplicitKillingSpinors}. Let $X_h$ be any Hamiltonian vector field on $S^2$, defined as in  \eqref{DefinitionHamiltonianField}. The following auxiliary result is proved in appendix \ref{ProofsChiralFermions}.
\begin{lemma}
\label{thm:AuxChiralFermionsS2}
For any complex function $f$ on the sphere, the derivatives $\cL_{X_h}\,\psi_{\frac{1}{2}, \, \pm \frac{1}{2}}$ satisfy the identities
\begin{align}
\label{MatrixElements1}
\int_{S^2} \langle\,  \Gamma_K \, f\,  \psi_{\frac{1}{2}, \, \pm \frac{1}{2}} \, , \; \cL_{X_h}  \psi_{\frac{1}{2}, \, \pm \frac{1}{2}} \, \rangle \, \vol_{g} \; &= \; -\, \frac{i}{16 \pi} \int_{S^2} (\, \Delta \bar{f} \,\pm \, 2i \, \Lie_{\partial_\phi}\bar{f} \, ) \, h\; \vol_{g}   \linebr
\int_{S^2} \langle\,  \Gamma_K \, f\,  \psi_{\frac{1}{2}, \,  \pm \frac{1}{2}} \, , \; \cL_{X_h}  \psi_{\frac{1}{2}, \, \mp \frac{1}{2}} \, \rangle \, \vol_{g} \; &= \;   \pm\, \frac{1}{2\sqrt{6\pi}}\, \int_{S^2} \big\{\bar{f}, \, Y_{1,\, \mp 1}^0  \big\} \, h\; \vol_{g} \nonumber  \ .
\end{align}
\end{lemma}
Here $\Lie_{\partial_\phi}\bar{f}$ is the Lie derivative of the scalar function along the vector field $\partial_\phi$ on $S^2$. The Poisson bracket of real functions  \cite[sec. 18.3]{Cannas} is defined by 
\beq
\label{DefinitionPoissonBracket}
\{\, h_1, \, h_2 \, \} \ := \ \vol_g(X_{h_1}, X_{h_2}) \ . 
\eeq
It is extended to complex functions by $i$-linearity on both entries. The scalar spherical harmonics $Y^0_{1,\, \pm 1}$ are given by
\beq
\label{ExplicitSphericalHarmonics}
Y_{1,\, -1}^0  \; = \;  \sqrt{\frac{3}{2\pi}} \, \frac{\bar{z}}{1 + |z|^2} \; = \;   \sqrt{\frac{3}{8\pi}}\, (\sin \theta)\, e^{-i\phi} \; = \; -\, \overbar{Y_{1,\, 1}^0} \ ,
 \eeq
using both spherical and complex coordinates on the sphere. 

We want to calculate the right-hand side of \eqref{MatrixElements1} when $f$ is a general spherical harmonic $Y^0_{l,\, n}$. These are eigenfunctions of the Laplacian and of the Lie derivative along $\partial_{\phi}$,
\begin{align}
\label{SHEigenfunctions}
 \Delta Y^0_{l,\, n} \; &= \; - \, l \, (l+1)\, Y^0_{l,\, n} \linebr
 \Lie_{\partial_\phi} (Y^0_{l,\, n}) \; &= \;  i\, n \, Y^0_{l,\, n} \nonumber \ .
\end{align}
See, for example, formulae (9.99), (9.102) and (9.104) in \cite{Jackson}. Moreover, since the $Y^0_{l,\, n}$ form an orthonormal basis of the space of $L^2$-integrable complex functions on $S^2$, the Poisson bracket of two spherical harmonics can itself be written as a linear combination of spherical harmonics. For our purposes, we only need the formula 
\beq 
\label{PoissonBracketSH}
\{ Y_{l,\, n}^0\, , \, Y_{1, \, \pm 1 }^0 \} \; = \; \mp \, i\, \sqrt{  \frac{3}{8\pi}} \,\sqrt{ (l  \mp n)\, (l\pm n +1 )} \; Y_{l, \, n\pm 1}^0 \ .
\eeq
It is a special case of the expression above (10.18) in reference \cite{RS}\footnote{In this reference, the Hamiltonian field $X_h$ is defined with the opposite sign of \eqref{DefinitionHamiltonianField}. These signs cancel in the definition \eqref{DefinitionPoissonBracket} of Poisson bracket.}. The general formula was first derived in \cite{FK} using slightly different normalizations.

Combining \eqref{MatrixElements1} with \eqref{SHEigenfunctions} and \eqref{PoissonBracketSH}, one obtains that
\begin{align}  
\label{MatrixElements2}
\blangle\,  \Gamma_K \, Y_{l, \, n}^0 \,  \psi_{\frac{1}{2}, \, \pm \frac{1}{2}} \, , \; \cL_{X_h}  \psi_{\frac{1}{2}, \, \pm \frac{1}{2}} \, \brangle_{L^2} \; &= \; \frac{i}{16\, \pi}\, [\, l(l+1) \mp 2n \, ] \, \blangle \, Y_{l , \, n}^0 , \, h \, \brangle_{L^2} \linebr
\blangle\,  \Gamma_K \, Y_{l, \, n}^0\, \psi_{\frac{1}{2}, \, \pm \frac{1}{2}} \, , \; \cL_{X_h}  \psi_{\frac{1}{2},  \, \mp \frac{1}{2}} \, \brangle_{L^2} \; &= \; -\, \frac{i}{8\, \pi}\, \sqrt{(l \mp n) (l \pm n+1)} \;    \blangle \, Y_{l, \,  n \pm 1}^0 , \, h  \, \brangle_{L^2}     \nonumber \ .
\end{align}
These formulae are very useful because any spinor $\psi_{l, \,  n}$ can be written as a linear combination of products of the form $Y_{l, \, n}^0 \,  \psi_{\frac{1}{2}, \, \pm \frac{1}{2}}$. So the matrix elements \eqref{MatrixElements2} give us all we need to obtain \eqref{ChiralExampleS2}.
More precisely, combining formulae (2.6) and (3.153) of \cite{TC},  each spinor $\psi_{l, \,  n}$ can be decomposed as
\beq \label{DecompositionSpinorsS2}
\psi_{l, \,  n} \; = \;  \sqrt{\frac{2\pi(l- n)}{l}} \; Y_{l - \frac{1}{2},  \,  n + \frac{1}{2}}^0 \; \, \psi_{\frac{1}{2}, \, - \frac{1}{2}}  \ + \ \sqrt{\frac{2\pi(l+n)}{l}} \; Y_{l - \frac{1}{2},  \,  n - \frac{1}{2}}^0 \; \, \psi_{\frac{1}{2}, \, \frac{1}{2}}  \ .
\eeq
Combining \eqref{DecompositionSpinorsS2} with \eqref{MatrixElements2}, a straightforward calculation finally yields
\begin{proposition}
\label{thm: ChiralMatrixElementsS2}
Given a Hamiltonian vector field $X_h$ on $S^2$, the derivatives $\cL_{X_h}\,\psi_{\frac{1}{2}, \, \pm \frac{1}{2}}$ of the Killing spinors satisfy the integral identities
\beq
\label{MatrixElements3}
\blangle\, \psi_{l,\, n}^+ \, , \, \cL_{X_h} \psi_{\frac{1}{2},\, \pm \frac{1}{2}}^+  \, \brangle_{L^2}  \,  - \,   \blangle\, \psi_{l,\, n}^- \, , \, \cL_{X_h} \psi_{\frac{1}{2},\, \pm \frac{1}{2}}^-  \, \brangle_{L^2} \; = \;  c_{l,\, n}^\pm \; \blangle \, Y_{l - \frac{1}{2},  \,  n \mp \frac{1}{2}}^0 \, , \, h \, \brangle_{L^2} \ 
\eeq
with coefficients $c_{l,\, n}^\pm := \frac{i}{8}\, \sqrt{\frac{l\pm n}{2\pi\, l}} \, \left(l- \frac{3}{2} \right)  \left(l- \frac{1}{2} \right)$.
\end{proposition}
This is the final expression in this section. It shows that the massive 4D gauge field linked to $X_h$ will generically have chiral interactions with the lightest 4D fermion, i.e. with the fermion $\varphi$ appearing in the higher-dimensional spinor $\varphi \otimes  \psi_{\frac{1}{2}, \, \frac{1}{2}}$. The amount of chirality depends on $X_h$ through the harmonic components of the Hamiltonian function $h$. 
The Killing choice $X_h = \partial_\phi$ corresponds to a choice of Hamiltonian function proportional to the spherical harmonic $Y^0_{1,0}$. Since spherical harmonics with different indices are $L^2$-orthogonal, the integral on the right-hand side of \eqref{MatrixElements3} will vanish unless $l=3/2$ and $n= 1/2$. But in this case the whole right-hand side of \eqref{MatrixElements3} still vanishes. Thus, in the Killing case, the derivative $\cL_X$ does not have chiral interactions with $\psi_{\frac{1}{2}, \, \frac{1}{2}}$, as expected.

\section{Conserved currents}
\label{ConservedCurrents}

Let be $\Psi$ and $\Phi$ be two complex spinors on $P= M_4\times K$ equipped with a Lorentzian, submersion metric $g_P \simeq (g_M, A, g_K )$. General properties of the Dirac operator say that
\beq
\label{AntiSelfAdjointDirac}
\langle \, \sD^P \Phi, \, \Psi \,  \rangle_{\Gamma_0} \: + \; \langle \, \Phi, \, \sD^P \Psi \,  \rangle_{\Gamma_0} \; = \;  \divergence_{g_\PPP} [\,  j_P  (\Phi ,\Psi) \, ] \ .
\eeq
Here  $\langle\cdot ,  \cdot \rangle_{\Gamma_0}$ is the pairing \eqref{Lorentzian pairing} while $j_P$ is the complex vector field on $P$ defined by
\beq
\label{DefHDCurrents}
j_P (\Phi ,\Psi) \; := \;  g_P^{r s} \,   \langle \, \Phi, \, u_r \cdot \Psi \,  \rangle_{\Gamma_0} \,  \, u_s   \ ,
\eeq
where $\{ u_r\}$ is a local, $g_P$-orthonormal trivialization of $TP$. In the Riemannian case, formula \eqref{AntiSelfAdjointDirac} is derived in \cite[sec. 2.3.4]{Bourguignon}, for example, and the Lorentzian case is similar. So when $\Phi$ and $\Psi$ are in the kernel of $\sD^P$, the vector field $j_P (\Phi ,\Psi)$ has vanishing divergence on $P$.
Note also that we have the general identity $ j_P (\Gamma_P \Phi ,\Gamma_P \Psi) =  j_P (\Phi ,\Psi)$ and that, for the particular choices $\Phi = \Psi$ and $\Phi = \Gamma_P \Psi$, the respective $j_P$ are real vector fields on $P$. All this follows from definition \eqref{DefHDCurrents} and the algebraic properties stated in section \ref{GammaMatrices}, such as  \eqref{SquaresChiralOperators}, \eqref{HermitianConjChiralOperator}, and the fact that Clifford multiplication by a vector is self-adjoint with respect to the $ \langle \cdot , \cdot  \rangle_{\Gamma_0}$ pairing on a Lorentzian manifold.

Now let $\{ X_\mu \}$ be a $g_M$-orthonormal trivialization of $TM_4$ and denote by $X_\mu^\HH$ the respective basic lifts to $P$. Define a vector field on $M_4$ through the fibre integrals
\beq
j_M (\Phi ,\Psi) \; := \;  g_M^{\mu \nu} \, \Big( \int_K \langle \, \Phi, \, X_\mu^\HH \cdot \Psi \,  \rangle_{\Gamma_0} \, \vol_{g_\KKK}  \Big ) \: X_\nu \ .
\eeq
Then $j_M$ is a section of $TM_4 \otimes \CC$. Its divergence with respect to $g_M$ is nicely related to the higher-dimensional divergence of $j_P$.
\begin{proposition}
\label{SpinorCurrents}
Let $\Phi$ and $\Psi$ be any two spinors on $P$ equipped with a Lorentzian, submersion metric $g_P \simeq (g_M, A, g_K )$. Then the divergence of the four-dimensional vector field $j_M (\Phi ,\Psi)$ is given by the fibre-integral
\beq
\divergence_{g_\MMM} [\, j_M (\Phi ,\Psi) \,  ] \; = \; \int_K \divergence_{g_\PPP} [\, j_P (\Phi ,\Psi) \,  ] \, \vol_{g_\KKK} \ .
\eeq
In particular, when $\Phi$ and $\Psi$ are in the kernel of $\sD^P$, it follows from \eqref{AntiSelfAdjointDirac} that the divergence of $j_M (\Phi ,\Psi)$ vanishes on $M_4$.
\end{proposition}
This result is a special case of proposition \ref{thm: FibreIntegralVectorField}, proved in the appendices. Now let $\xi$ be a Killing vector field on $(P, g_P)$. As was seen in \eqref{CommutatorDiracKLD}, the Kosmann-Lichnerowicz derivative $\cL_\xi$ commutes with the Dirac operator $\sD^P$. Thus, if $\sD^P \Psi = 0$, the spinor  $\Phi = \cL_\xi \Psi$ is also in the kernel of $\sD^P$. So we get the first part of:
\begin{corollary}
\label{thm:ChargeCurrents}
Let $\Psi$ be a complex spinor in the kernel of $\sD^P$ and let $\xi$ be a Killing vector field on $P$. Then the vector field $j_M (\cL_\xi \Psi, \Psi)$ on $M_4$
has vanishing divergence  with respect  to $g_M$. When $\xi$ is vertical on $P$, the vector field $j_M$ is purely imaginary. 
\end{corollary}
The last assertion is proved in appendix \ref{SectionFibreProjection}. It uses the following result:
\begin{lemma}[\!\!\cite{Bap2}, lemma A.3]
\label{KillingConditionsP}
Let $g_P \simeq (g_M, A, g_K)$ be a submersion metric on $P$. Then a vertical vector field $\xi$ is Killing with respect to $g_P$ if and only if its restriction to each fibre is Killing with respect to $g_K$ and, additionally, the Lie bracket $[\xi, X^\HH]$ vanishes for every vector field $X$ on $M_4$.
\end{lemma}

This lemma shows that finding Killing vector fields on $P$ is not easy when the submersion geometry is non-trivial. 
When the gauge fields are turned on, even if the internal metric $g_K$ remains constant throughout $M_4$, the Killing fields of $g_K$ may no longer be Killing fields of $g_P$. If $\xi$ is a Killing field of $g_K$ and we regard it as vertical vector field on $P$ that is constant along $M_4$, its Lie bracket with a basic field is
\beq
\label{KillingPCondition2}
[\xi, X^\HH] \; = \;  [\xi, X] \; + \; [\xi, A^a(X) \, e_a] \; = \;   A^a(X) \, [\xi, e_a] \ .
\eeq
The right-hand side does not vanish in general. So lemma \ref{KillingConditionsP} says that $\xi$ is not Killing on $P$ and we do not get conserved currents associated with $\xi$. 
A nice exception occurs if we assume that the geometry of $P$ satisfies the following conditions: 
\begin{enumerate}[itemsep=1mm, topsep=7pt]
\item[{\it i)}] the  internal metric $g_K$ is constant along $M_4$; 
\item[{\it ii)}]  the gauge one-form $A(X)$ has values in the space of Killing vector fields on $K$; 
\item[{\it iii)}]  $\xi$ is a Killing field of $g_K$ that commutes with all other Killing fields on $K$.
\end{enumerate}
Then clearly \eqref{KillingPCondition2} vanishes for all $X$. So $\xi$ will be Killing on $P$  and will define a non-trivial, conserved, charge current in 4D. Physically, the conditions $i$--$iii)$ say that we are in regions of spacetime where the Higgs-like scalars are constant; the non-zero gauge fields are all massless; and the internal Killing field $\xi$ is linked to an abelian, massless, electromagnetic-like gauge field on $M_4$. In this physical setting there are no background fields with $\xi$-charged bosons. So it is natural to find a conserved 4D fermionic current corresponding to the $\xi$-charge. For more perspective on this interpretation please compare with the geodesic picture described in \cite[sec. 5]{Bap2}.

\vspace{.2cm}

\begin{remark}
The example of a conserved current associated with $\cL_\xi$ for a Killing $\xi$ can be generalized. Let $S$ and $R$ be linear operators on spinors such that $\sD^P S \, \Psi =  R \, \sD^P \Psi$ for all $\Psi$. Then if $\Psi$ is in the kernel of $\sD^P$, so is  $\Phi = S \Psi$. Thus, proposition \ref{SpinorCurrents} implies that the 4D complex vector field $j_M (S \Psi, \Psi)$ has vanishing divergence with respect  to $g_M$. So a Kaluza-Klein model can have additional 4D conserved currents besides those defined by Killing fields. This will happen if the metric $g_P$ supports conformal Killing vector fields or conformal Killing-Yano forms, for example. The symmetries of the Dirac operator are a well-studied, beautiful topic. See, for instance, \cite{BK, Cariglia} for general accounts.
\end{remark}

\vspace*{1cm}

\section*{Acknowledgements}

\vspace*{-.1cm}

It is a pleasure to thank Nuno Romão for comments on an earlier version of the paper.

\newpage

\addtocontents{toc}{\cftpagenumbersoff{section}} 

\appendix
\appendixpage
\noappendicestocpagenum
\addappheadtotoc

\addtocontents{toc}{\cftpagenumberson{section}}

\section{Fibre-integral projection of a vector field}
\label{SectionFibreProjection}

Let $W$ be a vector field on $P = M\times K$ with compact support on the fibres. It determines a vector field on the base $M$ through the fibre-integral
\beq
\hat{W}_M  \; := \;  g_M^{\mu \nu} \, \Big( \int_K g_P(W, X_\mu^\HH)  \, \vol_{g_\KKK}  \Big)  \, X_\nu  \ .
\eeq
Here the $X_\mu$ form a $g_M$-orthonormal trivialization of $TM$ and the $X_\mu^\HH$ are their basic lifts to $P$, as in \eqref{BasicLiftX}.
When $W$ is itself a basic vector field on $P$, its projection to the base $\pi_\ast W$ is well-defined and the fibre-integrals simplify to yield $\hat{W}_M = (\Vol_{g_\KKK}) \, \pi_\ast W$. In the general case $\hat{W}_M$ is a more complicated average, but the following identity always holds.
\begin{proposition}
\label{thm: FibreIntegralVectorField}
Let  $g_P \simeq (g_M, A, g_K )$ be a submersion metric and $W$ a vector field on $P$. Then
\beq
\int_K \divergence_{g_\PPP} (W) \,  \vol_{g_\KKK} \; = \; \divergence_{g_\MMM} (\hat{W}_M)
\eeq
as functions on the base $M$.
\end{proposition}
\begin{proof}
Choose a local, $g_P$-orthonormal trivialization $\{ u_r\}$ of $TP$ adapted to the submersion, i.e. a trivialization of the form $\{X_\mu^\HH, v_j \}$, where the $X_\mu$ form a $g_M$-orthonormal trivialization of $TM$ and the $v_j$ form a $g_K$-orthonormal trivialization of $TK$.
On a Riemannian submersion the inner-product $g_P(X_\mu^\HH, X_\nu^\HH)$ is constant along the fibres and equal to $g_M(X_\mu, X_\nu)$. So we can decompose
\beq
\label{AuxDecompositionB1}
W \; = \; W^\VV \, +\, W^\HH \; = \; W^\VV \, +\,  \alpha^\mu \, X_\mu^\HH
\eeq
with coefficient functions $ \alpha^\mu := g_M^{\mu \nu}\, g_P(W, X_\nu^\HH)$ that vary both along the base and the fibres. This leads to a decomposition of the divergence into four terms:
 \begin{align}
 \divergence_{g_\PPP} (W) &:=  \; g_P^{rs}\,  \, g_P(\nabla_{u_r} W, u_s)   \nonumber \linebr
 \begin{split}
 &\, = \;  g_K^{ij}\,  \, g_P(\nabla_{v_i} W^\VV, v_j) \; +  \;   g_K^{ij}\,  \, g_P(\nabla_{v_i} W^\HH, v_j)  \linebr 
 &\ \ \ \ \,   + \;  g_M^{\mu \nu }\,  \, g_P(\nabla_{X^\HH_\mu} W^\VV, X^\HH_\nu) \;  +  \;  g_M^{\mu \nu }\,  \, g_P(\nabla_{X^\HH_\mu} W^\HH, X^\HH_\nu) \ .
 \end{split}
 \end{align}
 
 We will analyze each term separately. Using \eqref{IdentitySubmersionH2}, we have that
  \begin{align}
 g_K^{ij}\, g_P(\nabla_{v_i} W^\VV, v_j) \; &= \;  g_K^{ij}\, g_K(\nabla^K_{v_i} W^\VV, v_j) \; =\;   \divergence_{g_\KKK} (W^\VV) \ .
 \end{align}
 Using \eqref{IdentityCurvature2} and the fact that the curvature $F_A^a$ is anti-symmetric,
  \begin{align}
g_M^{\mu \nu }\,  \, g_P(\nabla_{X^\HH_\mu} W^\VV, X^\HH_\nu) \; &= \; - \, g_M^{\mu \nu }\,  \, g_P( W^\VV,  \nabla_{X^\HH_\mu} X^\HH_\nu) \;  \nonumber \linebr
&=\;  - \,  \frac{1}{2}\,   g_M^{\mu \nu }\, F_A^a (X_\mu, X_\nu )\, g_K(e_a, W^\VV) \; = \; 0 \ .
 \end{align}
 Using \eqref{IdentitySecondFundamentalForm} and the calculation \eqref{MeanCurvatureVectorProduct} of the trace of the second fundamental form,
 \begin{align}
 g_K^{ij}\,\,  g_P(\nabla_{v_i} W^\HH, v_j) \; &= \; - \, g_K^{ij}\, g_P(W^\HH,  \nabla_{v_i}  v_j) \; =\;   \frac{1}{2}\, \alpha^\mu\, g_K^{ij}\, (\dd^A g_K)_{X_\mu}(v_i, v_j)  \nonumber \linebr
 &= \;   \alpha^\mu\,  \big[ \,  \partial_{X_\mu}   \log \big(\sqrt{|g_K |} \big)  +  A^a_\mu \, \, \divergence_{g_\KKK} (e_a)  \, \big]  \ .
 \end{align}
 Finally, using \eqref{IdentitySubmersionH1}, decomposition \eqref{AuxDecompositionB1} and the definition of basic lift $X^\HH$, 
 \begin{align}
 g_M^{\nu \sigma }\,  \, g_P(\nabla_{X^\HH_\nu} W^\HH, X^\HH_\sigma) \; &= \; \alpha^\mu \, g_M^{\nu \sigma }\,  \, g_M(\nabla^M_{X_\nu} X_\mu, X_\sigma)  \; +\;  g_M^{\nu \sigma }  \, \, (g_M)_{\mu \sigma} \, (\dd \alpha^\mu) (X^\HH_\nu)  \;  \nonumber \linebr 
 &= \;   \alpha^\mu \,  \, \divergence_{g_\MMM} X_\mu \; + \; (\dd \alpha^\mu) (X_\mu)   \; + \; A_\mu^a \, (\dd \alpha^\mu) (e_a)  \ .
 \end{align} 
Thus, the combination of all four terms is
 \beq
 \divergence_{g_\PPP} (W) \; =  \;   \divergence_{g_\KKK} ( \, W^\VV + A_\mu^a \, \alpha^\mu\,  e_a \, ) \:  + \: \alpha^\mu\, \partial_\mu \big[ \log \big(\sqrt{|g_K |}  \big) \big] \: + \: \partial_\mu \alpha^\mu  \: + \: \alpha^\mu \,  \, \divergence_{g_\MMM} X_\mu  \ . 
\eeq
 Defining the vertical field $U = W^\VV +  A_\mu^a \, \alpha^\mu\,  e_a$, it follows that
 \beq
 \label{DivergencePreIntegration}
  \divergence_{g_\PPP} (W) \, \vol_{g_\KKK} \; = \;  \Lie_U ( \vol_{g_\KKK}) \; +\;  \partial_\mu ( \alpha^\mu  \, \vol_{g_\KKK}  ) \; +\; \alpha^\mu \,  \, (\divergence_{g_\MMM} X_\mu) \, \vol_{g_\KKK} \ .
 \eeq
Now denote $\omega^\mu := \alpha^\mu  \, \vol_{g_\KKK}$. The fibre-integral $\int_K \omega^\mu$ is a function on the base $M$. So integrating \eqref{DivergencePreIntegration} over the fibre, using Stokes' theorem and changing the order of the fibre-integral and a directional derivative on the base, 
 \beq
  \int_K \divergence_{g_\PPP} (W) \, \vol_{g_\KKK} \; =\;    \partial_\mu  \int_K \omega^\mu   \;  +  \; ( \divergence_{g_\MMM} X_\mu \, ) \int_K \omega^\mu \; = \;  \divergence_{g_\MMM} \Big[ X_\mu \int_K \omega^\mu \Big] \ .
 \eeq
 This concludes the proof of proposition \ref{thm: FibreIntegralVectorField}. 
\end{proof}

\vspace{.2cm}

\begin{proof}[Proof of corollary \ref{thm:ChargeCurrents}]
The first assertion follows from proposition \ref{SpinorCurrents} with the choice $\Phi = \cL_W \Psi$. For the second assertion, note that \eqref{CompatibilityKLDInnerProduct}, remark \ref{PairingRemark} and relation \eqref{CommutatorCliffordMultiplicationKLD} applied to the Killing case  imply that
\begin{align}
 \langle \,\cL_W \Psi, \, X_\mu^\HH \cdot  \Psi \,  \rangle_{\Gamma_0} \, &= \, -   \langle \, \Psi, \,  \cL_W \, (X_\mu^\HH \cdot  \Psi) \,  \rangle_{\Gamma_0} +  \Lie_W \langle \, \Psi, \,  X_\mu^\HH \cdot  \Psi \,  \rangle_{\Gamma_0}   \linebr
 &= \, - \, \langle \, \Psi, \,   X_\mu^\HH \cdot  \cL_W \, \Psi \,  \rangle_{\Gamma_0}  -   \langle \, \Psi, \,   [W, X_\mu^\HH] \cdot  \Psi \,  \rangle_{\Gamma_0}   +  \Lie_W \langle \, \Psi, \,  X_\mu^\HH \cdot  \Psi \,  \rangle_{\Gamma_0} . \nonumber
\end{align}
Now, if $W$ is Killing and vertical on $P$, lemma A.3 in \cite{Bap2} implies that $[W, X_\mu^\HH]$ is always zero and the restriction of $W$ to the fibres is Killing with respect to $g_K$. In particular $\divergence_{g_\KKK} W = 0$. So it follows from Stokes' theorem and \eqref{AdjointnessCliffordPairing} that
\begin{align}
\int_K  \langle \,\cL_W \Psi, \, X_\mu^\HH \cdot  \Psi \,  \rangle_{\Gamma_0} \,\vol_{g_\KKK} \;&= \; - \, \int_K  \langle \, \Psi, \,   X_\mu^\HH \cdot  \cL_W  \Psi \,  \rangle_{\Gamma_0} \,\vol_{g_\KKK} \linebr
&=  \; -  \, \int_K  \langle \,  X_\mu^\HH \cdot \Psi, \,   \cL_W  \Psi \,  \rangle_{\Gamma_0} \,\vol_{g_\KKK}   \nonumber  \linebr
&=  \; -  \, \int_K  \overbar{\langle \,\cL_W \Psi, \, X_\mu^\HH \cdot  \Psi \,  \rangle}_{\Gamma_0} \,\vol_{g_\KKK}    \nonumber
\end{align}
is a purely imaginary number.
\end{proof}

\section{Proofs for section \ref{SectionDecompositionDiracOperator}}
\label{DecompositionHDOperators}

This appendix contains proofs of propositions in section \ref{SectionDecompositionDiracOperator}. It deals with the decomposition of the higher-dimensional spinor connection and Dirac operator on a Riemannian submersion $P= M \times K$. The notation is as in sections   \ref{Submersions} and \ref{SectionDecompositionDiracOperator}.

\begin{proof}[Proof of proposition \ref{thm:DecompositionCovariantDerivatives}]
\label{proof:some-theorem}
First consider the case of a vertical vector field $U$. In the local trivialization $\{ X_\mu^\HH, v_j \}$ of $TP$, formula \eqref{StandardSpinorCovariantDerivative} decomposes as
\begin{multline}
\label{VerticalSplit}
\nabla_{U} \Psi \; =\; \partial_{U} \Psi \; +\;  \frac{1}{4} \, g_M^{\mu \sigma} \, g_M^{\nu \rho} \, \,  g_P(\nabla_{U} X_\mu^\HH, X_\nu^\HH) \, X_\sigma^\HH \cdot X_\rho^\HH \cdot \Psi \linebr
+ \; \frac{1}{4} \, g_K^{i j} \, g_K^{r s} \,  g_P(\nabla_{U} v_i, v_r) \, v_j \cdot v_s \cdot \Psi \; +\;  \frac{1}{2} \,  g_M^{\mu \nu} \, g_K^{i j} \,  g_P(\nabla_{U} X_\mu^\HH, v_i)  \, X_\nu^\HH \cdot v_j \cdot \Psi   \, .
\end{multline}
We want to simplify this expression when the higher-dimensional spinor is of the form $\Psi = \varphi^\HH(x) \otimes \psi(x,y)$, as in \eqref{TensorHDSpinor}. Since $U$ is vertical and $\varphi$ does not depend on the coordinate on $K$, 
\beq
\partial_{U} (\varphi^\HH \otimes \psi) \; = \; \varphi^\HH \otimes (\partial_{U} \psi) \ .
\eeq
Using that $\nabla$ has no torsion, that $[U, X^\HH]$ is a vertical vector field,  that $g_P(U, X_\mu^\HH)$ vanishes, that $g_P( X_\mu^\HH, X_\nu^\HH)$ is constant along the fibres, and also formula \eqref{IdentityCurvature2}, we get 
\begin{align}
g_P(\nabla_{U} X_\mu^\HH, \, X_\nu^\HH) \; &= \; g_P(\nabla_{X_\mu^\HH} U, \, X_\nu^\HH)  \;= \;  - \,  g_P(U,\,  \nabla_{X_\mu^\HH} X_\nu^\HH)  \nonumber \linebr
&= \; - \, \frac{1}{2} \, \,  g_P(F_A(X_\mu, X_\nu) ,  \, U) \ .
\end{align}
Using that $\nabla$ is $g_P$-compatible, that $g_P (X_\mu^\HH, v_i)$ vanishes, and identity \eqref{IdentitySecondFundamentalForm}, we get 
\begin{align}
g_P(\nabla_{U} X_\mu^\HH, v_i) \; &= \; \Lie_U[\, g_P (X_\mu^\HH, v_i) \, ]  \; - \; g_P(X_\mu^\HH,  \nabla_U v_i)   \linebr
&= \;  -\,  g_P(X_\mu^\HH,  \nabla_U v_i) \; = \;  \frac{1}{2} \, \,  (\dd^A g_K)_{X_\mu} (U, v_i) \ . \nonumber
\end{align}
Combining these formulae and using \eqref{EquivarianceClifford2}, decomposition \eqref{VerticalSplit} becomes \eqref{VerticalCovariantDerivative}, as desired.

Next, consider the case of a basic vector field $Y^\HH$.  In the local trivialization $\{ X_\mu^\HH, v_j \}$ of $TP$, formula \eqref{StandardSpinorCovariantDerivative} decomposes as
\begin{multline}
\label{HorizontalSplit}
\nabla_{Y^\HH} \Psi \; = \; \partial_{Y^\HH} \Psi \; + \; \frac{1}{4} \, g_M^{\mu \sigma} \, g_M^{\nu \rho} \,\,  g_P(\nabla_{Y^\HH} X_\mu^\HH, X_\nu^\HH) \, X_\sigma^\HH \cdot X_\rho^\HH \cdot \Psi \linebr
+ \; \frac{1}{4} \, g_K^{i r} \, g_K^{j s} \, \,  g_P(\nabla_{Y^\HH} v_i, v_j) \, v_r \cdot v_s \cdot \Psi \; + \;  \frac{1}{2} \, g_M^{\mu \nu} \, g_K^{i j} \, \,  g_P(\nabla_{Y^\HH} X_\mu^\HH, v_i)  \, X_\nu^\HH \cdot v_j \cdot \Psi  \ .
\end{multline}
Take a spinor on $P$ of the form $\Psi = \varphi^\HH(x) \otimes \psi(x,y)$, as in \eqref{TensorHDSpinor}. First, from the definition \eqref{BasicLiftX} of $Y^\HH$, the directional derivative of the local functions representing the spinors is
\beq
\label{Aux1}
\partial_{Y^\HH} (\varphi^\HH \otimes \psi) \; = \; (\partial_{Y}\varphi)^\HH \otimes \psi \; + \; \varphi^\HH \otimes \big[\partial_{Y} \psi + A^a(Y)\,  \partial_{e_a} \psi \big] \ .
\eeq
Secondly, using \eqref{IdentitySubmersionH1} and \eqref{EquivarianceClifford2},
\beq
\label{Aux2}
g_P(\nabla_{Y^\HH} X_\mu^\HH, X_\nu^\HH) \,  X_\sigma^\HH \cdot X_\rho^\HH \cdot (\varphi^\HH \otimes \psi) \; = \; g_M(\nabla_{Y} X_\mu, X_\nu) \,  (X_\sigma \cdot X_\rho \cdot \varphi)^\HH \otimes \psi \ .
\eeq
Thirdly, using \eqref{IdentityCurvature2} and \eqref{EquivarianceClifford2}, the last term in \eqref{HorizontalSplit} becomes
\begin{multline}
\label{Aux3}
\frac{1}{4} \, g_M^{\mu \nu} \, g_K^{i j} \, \, F^a_A (Y, X_\mu) \, \, g_P (e_a, \, v_i) \;   (  X_\nu \cdot \gamma_5 \cdot \varphi  )^\HH \otimes ( v_j \cdot \psi)  \linebr = \;  \frac{1}{4} \, g_M^{\mu \nu} \, \,  F^a_A (Y, X_\mu)  \;  (  X_\nu \cdot \gamma_5 \cdot \varphi  )^\HH \otimes (e_a \cdot \psi) \ . 
\end{multline}
Fourthly, from relation \eqref{EquivarianceClifford2} between Clifford multiplication on $S_\CC(P)$ and $S_\CC (\VV)$,
\beq
\label{Aux4}
g_K^{i r} \, g_K^{j s} \, \, g_P(\nabla_{Y^\HH} v_i, v_j) \, v_r \cdot v_s \cdot (\varphi^\HH \otimes \psi) \; = \; g_K^{i r} \, g_K^{j s} \, \, g_P(\nabla_{Y^\HH} v_i, v_j)\; \varphi^\HH \otimes (v_r \cdot v_s \cdot  \psi) \ .
\eeq
This last expression can be further developed to make the gauge fields $A^a(Y)$ appear explicitly. Using that $\nabla$ has no torsion, the definition \eqref{BasicLiftX} of horizontal lift, and the fact that the $A^a(Y)$ are constant functions along the fibres, we have
\begin{align}
g_P(\nabla_{Y^\HH} v_i, v_j) \; &= \; g_P(\, \nabla_{v_i} Y^\HH  +  [Y^\HH, v_i] , \, v_j \, )  \linebr
&= \; - \, g_P(\, Y^\HH  , \, \nabla_{v_i}  v_j \, ) \; + \; g_P(\, [Y, v_i] + A^a(Y) [e_a, v_i]   ,\, v_j \, )  & \nonumber \linebr
&= \; \frac{1}{2} (\dd^A g_K)_Y (v_i, v_j) \; + \; g_P(\, [Y, v_i], v_j) \; + \; A^a(Y) \,  g_K ( [e_a, v_i]   , v_j ) \nonumber \ .
\end{align}
The last equality also uses identity \eqref{IdentitySecondFundamentalForm} for the second fundamental form of the fibres.
Since the connection $\nabla$ is metric, $g_P(\nabla_{Y^\HH} v_i, v_j)$ is anti-symmetric in the indices $i$ and $j$. In contrast, the first term in the last line is symmetric. So we can eliminate it through
\begin{align}
\label{Aux5}
2\, g_P(\nabla_{Y^\HH} v_i, v_j) \; =& \; g_P(\nabla_{Y^\HH} v_i, v_j) - g_P(\nabla_{Y^\HH} v_j,  v_i)   \linebr
\begin{split}
={}& \;  g_P(\, [Y, v_i], v_j) \, - \, g_P(\, [Y, v_j ] , v_i) \, \linebr
 & + \;  A^a(Y) \big\{ \,  g_K ( [e_a, v_i]   , v_j )  \,  - \, g_K ( [e_a, v_j]   , v_i )  \,   \big\} \ . \nonumber
\end{split}
\end{align}
Finally, substituting \eqref{Aux1}, \eqref{Aux2}, \eqref{Aux3} and \eqref{Aux4} into decomposition \eqref{HorizontalSplit}, using the standard definition of the covariant derivative $\nabla^M\varphi$ on $M$ (analogous to \eqref{StandardSpinorCovariantDerivative}), and using the definition of Kosmann-Lichnerowicz derivative in \eqref{DefinitionKLDerivative}, we get formula \eqref{HorizontalCovariantDerivative} for the higher-dimensional covariant derivative $\nabla_{Y^\HH} \Psi$, as desired.
\end{proof}

\vspace{.2cm}

\begin{proof}[Proof of proposition \ref{thm:DecompositionDiracOperator}]
We first calculate the vertical part of the higher-dimensional Dirac operator. It follows from \eqref{VerticalCovariantDerivative} and \eqref{EquivarianceClifford2} that
\begin{align}
\label{VerticalCovariantDerivative3}
v_j \cdot \nabla_{v_j} \Psi \; =& \; (\gamma_5 \cdot \varphi)^\HH \otimes(v_j \cdot \nabla_{v_j} \psi)   \linebr
&- \;   \frac{1}{8} \, \, g_M^{\mu \nu}\, g_M^{\sigma \rho}\, \, (F_A^a)_{\mu \sigma } \, \, g_K(e_a , v_j) \; (\gamma_5 \cdot X_\nu \cdot X_\rho \cdot \varphi)^\HH \otimes (v_j \cdot \psi)   \nonumber \linebr
&+ \, \frac{1}{4} \,  \, g_M^{\mu \nu}\,  \sum\nolimits_{i} \,  (\dd^A g_K)_{X_\mu} (v_j, v_i)\; \, (  \gamma_5 \cdot X_\nu \cdot \gamma_5 \cdot \varphi  )^\HH \otimes (v_j \cdot v_i \cdot \psi) \ .  \nonumber
\end{align}
Now, the covariant derivative $(\dd^A g_K)_{X_\mu} (v_j, v_i)$ is symmetric in $i$ and $j$, while the Clifford product $ v_j \cdot v_i$ is anti-symmetric for $i\neq j$. So a sum over $j$ simplifies to
\begin{align}
 \sum\nolimits_{ij} \,  (\dd^A g_K)_{X_\mu} (v_j, v_i)\  v_j \cdot v_i \cdot \psi \; &= \; \sum\nolimits_{j} \,  (\dd^A g_K)_{X_\mu} (v_j, v_j)\  v_j \cdot v_j  \cdot \psi \linebr
 &= \; -\,  \sum\nolimits_{j} \,  (\dd^A g_K)_{X_\mu} (v_j, v_j) \; \psi  \nonumber   \ .
 \end{align}
But formulae (2.25) and (2.28) in section 2.4 of \cite{Bap} say that the metric trace of $\dd^A g_K$ is
 \begin{align}
 \label{MeanCurvatureVectorProduct}
 \sum\nolimits_{j} (\dd^A g_K)_{X_\mu} (v_j, v_j) \; &= \; -  \, 2  \, \sum\nolimits_{j} g_P(\nabla_{v_j} v_j , \, X_\mu ^\HH) \;  \nonumber \linebr
 &= \; 2\, \big[ \,  \partial_{X_\mu}   \log \big(\sqrt{|g_K |} \big)  +  A^a_\mu \, \, \divergence (e_a)  \, \big] \ .
 \end{align}
 In geometric terms, the trace of the second fundamental form of the fibres is their mean curvature vector $N$. Then \eqref{MeanCurvatureVectorProduct} is just the inner-product $- 2 g_P(N, X_\mu^\HH)$ written in a way that makes explicit the dependence on the gauge fields, as in \cite{Bap}.
 
 Thus, summing \eqref{VerticalCovariantDerivative3} over $j$ and using that $\gamma_5\,\gamma_5 = 1$, the result is
 \begin{align}
 \label{VerticalCovariantDerivative4}
 \sum\nolimits_{j}  v_j \cdot \nabla_{v_j} \Psi \; =& \; (\gamma_5 \cdot \varphi)^\HH \otimes( \sD^K \psi)   \; - \;   \frac{1}{8} \, \, (F_A^a)^{\mu \nu} \; (\gamma_5 \cdot X_\mu \cdot X_\nu \cdot \varphi)^\HH \otimes (e_a \cdot \psi)     \linebr
&+\; \frac{1}{2} \,  \, g_M^{\mu \nu}\,  \big[ \, \partial_{X_\mu}   \log \big(\sqrt{|g_K |} \big)  +  A^a_\mu \, \, \divergence (e_a) \, \big]  \, (  X_\nu \cdot \varphi  )^\HH \otimes  \psi \ .  \nonumber 
  \end{align}
  
To calculate the horizontal part of $\sD^P$, take \eqref{HorizontalCovariantDerivative} and observe that
 \begin{align}
\label{HorizontalCovariantDerivative3}
X_\mu^\HH \cdot \nabla_{X_\nu^\HH} \Psi \; =& \;  (X_\mu \cdot \nabla^M_{X_\nu} \, \varphi)^\HH \otimes \psi \; + \;  A^a_\nu \, \, (X_\mu \cdot \varphi)^\HH \otimes  \cL_{e_a} \psi  \linebr 
\begin{split}
 &+ \, \frac{1}{4}\, \, g_M^{\sigma \rho}\, \, F_A^a(X_\nu , X_\sigma) \, \, ( X_\mu \cdot   X_\rho \cdot \gamma_5 \cdot \varphi )^\HH \otimes (e_a\cdot \psi)   \linebr
 &+\;  (X_\mu \cdot \varphi)^\HH \otimes \,  \partial_{X_\nu}  \psi  \  \nonumber \linebr
  &+\;  (X_\mu \cdot \varphi)^\HH \otimes \Big( \frac{1}{8} \, \sum\nolimits_{ij}   \big\{ g_P([X_\nu, v_i], v_j)  -  g_P([X_\nu, v_j ] , v_i) \big\}  \, v_i \cdot v_j \cdot \psi  \Big) \, . \nonumber
\end{split}
\end{align}
Contracting this expression with $g_M^{\mu \nu}$ and summing with \ref{VerticalCovariantDerivative4} leads directly to \ref{GeneralFormulaDecompositionDiracOperator}.
\end{proof}

\section{Proofs for sections \ref{ChiralFermions} and \ref{ExplicitExampleS2}}
\label{ProofsChiralFermions}

\begin{proof}[Proof of proposition \ref{EquivalentConditions}]
Assume condition \ref{it1}.
The chirality operator $\Gamma_K$ anti-commutes with $\sD$ and maps $V_{m>0} \rightarrow V_{m<0}$ isomorphically. Since $\cL_X$ commutes with $\Gamma_K$, it is clear that $\cL_X$ also preserves $V_{m<0}$. But if $\cL_X$ preserves $V_{m>0} \oplus V_{m<0}$, the inner-product 
\beq
\label{InnerProductIdentity1}
\blangle  \cL_X \phi, \psi \brangle_{L^2} \; = \; -\,  \blangle \phi, \cL_X \psi \brangle_{L^2}
\eeq
must vanish for all $\psi$ in $V_{m>0} \oplus V_{m<0}$ and all $\phi$ in $V_0$. This follows from the orthogonality of decompositions \eqref{DecompositionsSpinorSpaces2} and the fact that, due to \eqref{IntegralKLD},  $\cL_X$ is anti-self-adjoint with respect to the $L^2$-inner-product for a divergence-free $X$. So the spinor $\cL_X \phi$ does not have components in $V_{m>0} \oplus V_{m<0}$ for any $\phi$ in $V_0$.  We conclude that $\cL_X \phi$ must be in $V_0$ for any $\phi$ in $V_0$, and hence $\cL_X$ preserves $V_0$.
Moreover, since $\cL_X$ commutes with $\Gamma_K$, it must also preserve its eigenspaces $V_0^\pm$ inside $V_0$. Similarly, it must preserve the  two $\Gamma_K$-eigenspaces $V_{m>0}^\pm$ inside $V_{m>0} \oplus V_{m<0}$. This concludes the proof that condition \ref{it1} is equivalent to condition \ref{it2}.

The definition of the chiral projections $\psi^\pm$ in \eqref{DefChiralProjections} implies the identity
\beq
\label{InnerProductIdentity2}
\langle \phi^+\!, \cL_X \psi^+ \brangle_{L^2} \, - \, \langle \phi^- \!, \cL_X \psi^- \brangle_{L^2} \; = \; \blangle  \, \Gamma_K \phi, \, \cL_X \psi \,\brangle_{L^2} \ .
\eeq
Since $\Gamma_K$ preserves $V_0$ and maps $V_{m > 0}$ isomorphically onto $V_{m < 0}$, condition \ref{it4} can be rephrased as 

\vspace{.1cm}

\noindent
\text{\it \ref{it4}'. $\blangle  \, \phi, \, \cL_X \psi \,\brangle_{L^2} = \,  0$ for all $\psi$ in $V_{m>0}$ and all $\phi$ in $V_{m \leq 0}$.}

\vspace{.1cm}

\noindent
It says that $ \cL_X \psi$ is in the orthogonal complement of $V_{m \leq 0}$ for all $\psi$ in $V_{m>0}$. But that orthogonal complement is $V_{m>0}$ itself. So condition \ref{it4} is equivalent to condition \ref{it1}.

Now suppose that $\psi$ is an eigenspinor of $\sD$ with positive eigenvalue $m >0$ and that $\phi$ is an eigenspinor with non-positive eigenvalue $-\mu \leq 0$. Since $\sD$ is formally self-adjoint with respect to the $L^2$-inner-product of spinors, we have
\beq
\blangle\, \phi, [\cL_X, \sD]\, \psi \, \brangle_{L^2} \; = \;   ( \mu + m)\, \blangle \, \phi, \cL_X \psi \, \brangle_{L^2}  \ .
\eeq
The real factor $( \mu + m)$ is always positive, so $\blangle\, \phi, [\cL_X, \sD]  \psi \, \brangle_{L^2}$ vanishes if and only if $\blangle \, \phi, \cL_X \psi \, \brangle_{L^2}$ vanishes. So $[\cL_X, \sD]  \psi$ is orthogonal to $V_{m \leq 0}$ if and only if $\cL_X \psi$ is. In other words, $[\cL_X, \sD] \psi$ is in $V_{m > 0}$ if and only if $\cL_X \psi$ is in that space. Hence, condition \ref{it3} is equivalent to condition \ref{it1}.

To prove the last assertion in proposition \ref{EquivalentConditions}, assume that points \ref{SCSit1} and \ref{SCSit2} of definition \ref{DefStrongChiralSymmetry} are satisfied. Point \ref{SCSit1} almost implies the full condition \ref{it4} of proposition \ref{EquivalentConditions}. We only need to check that
$\blangle \phi^+, \cL_X \psi^+ \brangle_{L^2} = \blangle \phi^-, \cL_X \psi^- \brangle_{L^2}$ for all $\psi$ in $V_{m>0}$ and all $\phi$ in $V_0$. Using identity \eqref{InnerProductIdentity2} and the fact that $\Gamma_K \phi$ is in $V_0$ iff $\phi$ belongs to that space, that is equivalent to verifying that $\blangle  \,  \phi, \, \cL_X \psi \,\brangle_{L^2}$ vanishes for all $\psi$ in $V_{m>0}$ and $\phi$ in $V_{0}$. But point \ref{SCSit2} says that $\cL_X$ preserves $V_0$, so $\blangle  \cL_X \phi, \psi \brangle_{L^2}$ always vanishes. For a divergence-free $X$, identity \eqref{InnerProductIdentity1}  then guarantees that $\blangle  \,  \phi, \, \cL_X \psi \,\brangle_{L^2}$ also vanishes, as desired. Thus, points \ref{SCSit1} and \ref{SCSit2} of definition \ref{DefStrongChiralSymmetry} imply condition \ref{it4} of proposition \ref{EquivalentConditions}, and hence all of its conditions.

In the opposite direction, suppose that $\cL_X$ satisfies the four equivalent conditions of proposition \ref{EquivalentConditions}. Then, conditions \ref{it2} and \ref{it4} directly imply points \ref{SCSit2} and \ref{SCSit1} of the definition \ref{DefStrongChiralSymmetry}, respectively. So we conclude that the desired equivalence holds.
\end{proof}

\vspace{.2cm}

\begin{proof}[Proof of lemma \ref{thm:AuxChiralFermionsS2}]

Identify the punctured sphere $S^2  \setminus \! \{N\}$ with the complex plane through stereographic projection and let $z = x + i y$ be the standard complex coordinate on $\CC$. In this chart, the round metric on $S^2$ with total volume $4\pi$ takes the form
\beq
g \; = \; \frac{4}{q^2} \, \,  [ \, (\dd x)^2 \, + \, (\dd y)^2 \,] \ ,
\eeq
where $q := 1 + |z|^2$. Then the vector fields $v_1 = \frac{q}{2}\, \partial_x$ and $v_2 = \frac{q}{2}\, \partial_y$ define an oriented, $g$-orthonormal frame on the tangent spaces. The standard complex structure on $S^2$ acts on them as $J v_1 = v_2$. The coordinate vector field corresponding to the azimuthal angle $0 < \phi < 2\pi$ on the sphere is 
\beq
\partial_\phi \; = \;  - \, y \, \partial_x \: +\: x \, \partial_y\ .
\eeq
It is a Killing vector field with respect to $g$. The gamma matrices in two Euclidean dimensions can be taken to be 
\beq \label{PauliMatrices}
\gamma_1 \, = \, \bmatr 0  &  \ -i \\ -i  & \ 0 \ematr \qquad  \qquad \gamma_2 \, = \, \bmatr 0  &  \ -1 \\ 1  & \ 0 \ematr \  \ .
\eeq
They are anti-hermitian matrices satisfying $\{\gamma_a, \gamma_b \} = - \delta_{ab} \, I_2$. The chirality operator \eqref{ChiralOperatorRS} takes the form
\beq
\Gamma_K \; = \; i \, \gamma_1\, \gamma_2 \; = \;  \bmatr \,1  &  \, 0 \\ 0  & \, -1 \ematr  \ .
\eeq
A vector field $X = X^1 v_1 + X^2 v_2$ on $S^2$ acts on spinors as $X\cdot \psi = (X^1 \gamma_1 +  X^2 \gamma_2) \, \psi$. Then a direct computation using formula \eqref{ExplicitKillingSpinors} for the Killing spinors leads to the identities
\begin{align}
\label{Aux1IdentitiesS2}
\langle \psi_{\frac{1}{2}, \, \pm \frac{1}{2}} \, ,\, X \cdot   \psi_{\frac{1}{2}, \, \pm \frac{1}{2}} \rangle \; &= \; \mp \, \frac{i}{4\pi} \, \, g(X, \partial_\phi)    \linebr
\langle \Gamma_K \, \psi_{\frac{1}{2}, \, \pm \frac{1}{2}} \, ,\, X \cdot   \psi_{\frac{1}{2}, \, \pm \frac{1}{2}} \rangle \; &= \; \pm \, \frac{1}{4\pi} \,\, g(JX, \partial_\phi) \nonumber  \linebr
\langle \Gamma_K \, \psi_{\frac{1}{2}, \, \pm \frac{1}{2}} \, ,\, X \cdot   \psi_{\frac{1}{2}, \, \mp \frac{1}{2}} \rangle \; &= \; \mp \, \frac{1}{\sqrt{6\pi}} \, \, \dd (Y_{1,\, \mp 1}^0)(X) \nonumber  \ .
\end{align}
In the last equality, the functions $Y_{1,\, \pm 1}^0$ are the scalar spherical harmonics on $S^2$ written explicitly in \eqref{ExplicitSphericalHarmonics}.

Now, since the $\psi_{\frac{1}{2}, \, \pm \frac{1}{2}}$ are Killing spinors, identities \eqref{KillingSpinorsS2} allow us to simplify formula \eqref{KLDerivativeRS} for the Kosmann-Lichnerowicz derivative and write
\beq
\label{KLDerivativeKS}
\cL_X \, \psi_{\frac{1}{2}, \, \pm \frac{1}{2}}  \ = \  -\, \frac{1}{2} \, X \cdot \psi_{\frac{1}{2}, \, \pm \frac{1}{2}}  \; - \;  \frac{i}{4} \,\divergence(JX) \, \Gamma_K \, \psi_{\frac{1}{2}, \, \pm \frac{1}{2}}  \ .
\eeq
Combining this formula with \eqref{Aux1IdentitiesS2} yields
\begin{align}
\label{Aux2IdentitiesS2}
\langle \Gamma_K \, f\, \psi_{\frac{1}{2}, \, \pm \frac{1}{2}} \, ,\, \cL_X \, \psi_{\frac{1}{2}, \, \pm \frac{1}{2}} \rangle \; &= \;- \, \frac{1}{8\pi} \, \big[ \pm g(JX, \partial_\phi) \; +\; \frac{i}{2} \, \divergence(JX) \big]  \,\bar{f}  \linebr
\langle \Gamma_K \, f\, \psi_{\frac{1}{2}, \, \pm \frac{1}{2}} \, ,\, \cL_X \,  \psi_{\frac{1}{2}, \, \mp \frac{1}{2}} \rangle \; &= \; \pm \, \frac{1}{2 \sqrt{6\pi}} \, \bar{f} \, \, \dd (Y_{1,\, \mp 1}^0)(X) \ . \nonumber
\end{align}
Next, consider the case where $X_h$ is a Hamiltonian vector field on $S^2$ determined by a function $h$, as in \eqref{DefinitionHamiltonianField}. Then formula \eqref{IdentityHamiltonianFieldRS} leads to the first two equalities in
\begin{align}
\label{Aux3IdentitiesS2}
\divergence(JX_h) \; &= \; \Delta\, h \linebr
g(JX_h, \partial_\phi) \; &= \; g(\, \gradient \, h, \,\partial_\phi \, ) \; = \; \Lie_{\partial_\phi} \,  h   \nonumber    \linebr
\dd (Y_{1,\, \pm 1}^0)(X_h) \; &= \; \Lie_{X_h} \, Y_{1,\, \pm 1}^0 \; = \;  \{ Y_{1,\, \pm 1}^0 \, , \, h \} \nonumber \ .
\end{align}
The last equality follows from the definition of Poisson bracket (e.g. \cite[sec. 18.3]{Cannas}). Thus
\begin{align}
\label{Aux4IdentitiesS2}
\langle \Gamma_K \, f\, \psi_{\frac{1}{2}, \, \pm \frac{1}{2}} \, ,\, \cL_{X_h} \, \psi_{\frac{1}{2}, \, \pm \frac{1}{2}} \rangle \; &= \;- \, \frac{1}{8\pi} \, \big[ \pm \Lie_{\partial_\phi} \,  h  \; +\; \frac{i}{2} \, \Delta\, h \big]  \,\bar{f}  \linebr
\langle \Gamma_K \, f\, \psi_{\frac{1}{2}, \, \pm \frac{1}{2}} \, ,\, \cL_{X_h} \,  \psi_{\frac{1}{2}, \, \mp \frac{1}{2}} \rangle \; &= \; \pm \, \frac{1}{2 \sqrt{6\pi}} \, \bar{f} \, \, \{ Y_{1,\, \mp 1}^0 \, , \, h \} \ . \nonumber
\end{align}
Lemma \ref{thm:AuxChiralFermionsS2} is concerned with the integrals over $S^2$ of the matrix elements \eqref{Aux4IdentitiesS2}.
But the Laplacian is self-adjoint with respect to the $L^2$-product of functions, so
\beq
\label{IntegralIdentityAux1S2}
\int_{S^2} \bar{f} \, \, \Delta h \; \vol_{g} \; = \; \int_{S^2}   (\Delta \bar{f}) \, h \; \vol_{g} \ .
\eeq
Since $\partial_\phi$ is Killing with respect to $g$, it has zero divergence, so from Stokes' theorem:
\beq
\label{IntegralIdentityAux2S2}
\int_{S^2} \bar{f} \, \, (\Lie_{\partial_\phi} h) \; \vol_{g} \; = \; - \, \int_{S^2}   (\Lie_{\partial_\phi} \bar{f}) \, h \; \vol_{g} \ .
\eeq
Moreover, for any three functions the Poisson bracket satisfies the identity \cite[sec. 18.3]{Cannas}
\beq
\label{IdentityPoissonBracket}
 \{ f_1 \,  , \, f_2\, f_3 \} \;  = \;   \{ f_1 \,  , \, f_2 \} \, f_3  \; + \;  f_2\, \{ f_1 \,  , \, f_3 \} \ .
\eeq
Since Hamiltonian vector fields have zero divergence, the integral of the Poisson bracket of any two functions always vanishes. So the integral of the left-hand side of \eqref{IdentityPoissonBracket} vanishes, and, for the obvious choices of $f_j$, we get that
\beq
\label{IntegralIdentityAux3S2}
\int_{S^2} \bar{f} \, \, \{ Y_{1,\, \pm 1}^0 \, , \, h \} \; \vol_{g} \; = \;  - \, \int_{S^2}   \{ Y_{1,\, \pm 1}^0 \, , \, \bar{f} \} \, h \; \vol_{g}  \; = \;   \int_{S^2}   \{ \bar{f} \, , \,  Y_{1,\, \pm 1}^0 \} \, h \; \vol_{g} \ .
\eeq
Combining \eqref{Aux4IdentitiesS2} with the integral identities \eqref{IntegralIdentityAux1S2}, \eqref{IntegralIdentityAux2S2} and \eqref{IntegralIdentityAux3S2}, we finally get the formulae \eqref{MatrixElements1} in lemma \ref{thm:AuxChiralFermionsS2}.
\end{proof}

\newpage


\renewcommand{\baselinestretch}{1.2}

\vspace{1cm}

\end{document}